\documentclass[10pt,a4paper]{article}

\usepackage{amsmath, amsthm, amssymb}
\newtheorem{theorem}{Theorem}[section]

\newtheorem{lemma}[theorem]{Lemma}
\newtheorem{proposition}[theorem]{Proposition}

\newtheorem{corollary}[theorem]{Corollary}
\theoremstyle{definition}
\newtheorem{definition}{Definition}[section]
\theoremstyle{remark}
\newtheorem{example}{Example}[section]
\newtheorem{remark}{Remark}[section]

\usepackage{a4wide}
\usepackage{url}

\usepackage{graphicx}
\usepackage[utf8]{inputenc}
\usepackage[T1]{fontenc}
\usepackage{bussproofs}
\usepackage{tikz}
\usepackage{cmll}
\usepackage{wrapfig}

\newcommand{\Def}[1]{\textbf{#1}}

\newcommand{\Nat}{\mathbb{N}}
\newcommand{\Rpos}{\mathbb{R}_{\geq 0}}
\newcommand{\Var}{\mathcal{V}}
 \newcommand{\srt}{\Delta}     
 \newcommand{\srpt}{!\Delta}   
 \newcommand{\gsrt}{(!)\Delta} 
 \newcommand{\psrt}{\srt^\oplus}
 \newcommand{\psrpt}{\srpt^\oplus}
 \newcommand{\pgsrt}{\gsrt^\oplus}
 \newcommand{\rt}{\Rpos^{\srt}}
 \newcommand{\rpt}{\Rpos^{\psrpt}}
 \newcommand{\grt}{\Rpos^{\gsrt}}
 \newcommand{\prt}{\Rpos^{\psrt}}
 \newcommand{\prpt}{\Rpos^{\psrpt}}
 \newcommand{\pgrt}{\Rpos^{\pgsrt}}
 \newcommand{\fprt}{\Rpos^{(\psrt)}}
 \newcommand{\fprpt}{\Rpos^{(\psrpt)}}
 \newcommand{\fpgrt}{\Rpos^{(\pgsrt)}}
\newcommand{\hnf}{\mathrm{hnf}}
\newcommand{\PBT}{\mathcal{PT}}
\newcommand{\VPBT}{\mathcal{VT}}

\newcommand{\rTsupp}[1]{\mathcal{T}^r(#1)}

\newcommand{\Mf}[1]{\mathrm{M}_\mathrm{fin}(#1)}
\newcommand{\Mfp}[2]{\mathrm{M}^{#2}_\mathrm{fin}(#1)}
\newcommand{\Dom}[1]{\mathrm{Dom}(#1)}
\newcommand{\Iso}[1]{\mathrm{Iso}(#1)}

\newcommand{\supp}[1]{\mathrm{supp}(#1)}

\newcommand{\linapp}[2]{\langle#1\rangle\ #2}
\newcommand{\bag}[1]{\overline{#1}}
\newcommand{\ls}[2]{#2\oplus_{#1}\cdot}
\newcommand{\rs}[2]{\cdot\oplus_{#1}#2}
\newcommand{\llabel}{\mathrm l}
\newcommand{\rlabel}{\mathrm r}
\newcommand{\preflist}[2]{#1 \cdot #2}

\newcommand{\rbeta}{\rightarrow_\beta}
\newcommand{\rsum}{\rightarrow_\oplus}
\newcommand{\rightarrowdbl}{\rightarrow\mathrel{\mkern-14mu}\rightarrow}
\newcommand{\reduc}[3]{#1 \vdash #2 \rightarrowdbl #3}

\newcommand{\Lred}[2][]{\mathrm L^{#1}(#2)}
\newcommand{\forget}[1]{|#1|}
\newcommand{\prob}[1]{\mathcal{P}(#1)}
\newcommand{\mcoef}[1]{\mathrm m(#1)}
\newcommand{\rtay}[1]{#1^{*\oplus}}
\newcommand{\tay}[1]{#1^*}
\newcommand{\nf}[1]{\mathrm{nf}(#1)}
\newcommand{\bt}{\mathit{BT}}
\newcommand{\pbt}{\mathit{PT}}
\newcommand{\vpbt}{\mathit{VT}}

\newcommand{\hprob}[2]{\mathcal{P}\left(#1 \rightarrowdbl #2\right)}
\newcommand{\conv}[1]{\mathcal{P}_\Downarrow(#1)}
\newcommand{\size}[1]{||#1||}
\newcommand{\ssize}[1]{||#1||^\dagger}

\newcommand{\subst}[2]{\left[\raisebox{.2em}{$#1$}/\raisebox{-.2em}{$#2$}\right]}
\newcommand{\dsubst}[3]{\delta_#3#1\cdot#2}

\newcommand{\bareq}{\mathrel{\equiv_{\mathrm{bar}}}}
\newcommand{\sep}{\mathrel{|}}

\newcommand{\pbtone}{\mathbf T}
\newcommand{\vpbtone}{\mathbf t}
\newcommand{\bttone}{T}
\newcommand{\btttwo}{U}
\newcommand{\bhtone}{t}
\newcommand{\bhttwo}{u}

\newenvironment{varitemize}
{
\begin{list}{\labelitemi}
{
\setlength{\itemsep}{0pt}
 \setlength{\topsep}{0pt}
 \setlength{\parsep}{0pt}
 \setlength{\partopsep}{0pt}
 \setlength{\leftmargin}{15pt}
 \setlength{\rightmargin}{0pt}
 \setlength{\itemindent}{0pt}
 \setlength{\labelsep}{5pt}
 \setlength{\labelwidth}{10pt}}}
{
 \end{list} 
}

\title{On the Taylor Expansion of Probabilistic $\lambda$-terms\\ (Long Version)}

\author{Ugo Dal Lago \and Thomas Leventis}\date{}

\begin{document}

\maketitle

\begin{abstract}
  We generalise Ehrhard and Regnier's Taylor expansion from
  \emph{pure}  to \emph{probabilistic} $\lambda$-terms through
  notions of probabilistic resource terms and explicit Taylor
  expansion. We prove that the Taylor expansion is adequate when seen
  as a way to give semantics to probabilistic $\lambda$-terms, and that
  there is a precise correspondence with probabilistic B\"ohm trees, as
  introduced by the second author.
\end{abstract}

\section{Introduction}

Linear logic is a proof-theoretical framework which, since its
inception~\cite{G87}, has been built around an analogy between on the
one hand linearity in the sense of linear algebra, and on the other
hand the absence of copying and erasing in cut elimination and
higher-order rewriting. This analogy has been pushed forward by
Ehrhard and Regnier, who introduced a series of logical and
computational frameworks accounting, along the same analogy, for
concepts like that of a differential, or the very related one of an
approximation. We are implicitly referring to differential
$\lambda$-calculus~\cite{ER03}, to differential linear
logic~\cite{ER06}, and to the Taylor expansion of ordinary
$\lambda$-terms~\cite{ER08}.  The latter has given rise to an
extremely interesting research line, with many deep contributions in
the last ten years. Not only the Taylor expansion of
pure $\lambda$-terms has been shown to be endowed with a well-behaved
notion of reduction, but the B\"ohm tree and Taylor expansion
operators are now known to commute~\cite{ER06b}.
This easily implies that the equational theory (on pure
$\lambda$-terms) induced by the Taylor expansion coincides with the
one induced by B\"ohm trees.

The Taylor expansion operator is essentially \emph{quantitative}, in
that its codomain is not merely the set of resource
$\lambda$-terms~\cite{B93,ER03}, a term syntax for promotion-free
differential proofs, but the set of \emph{linear combinations} of
those terms, with positive real number coefficients. When
enlarging the domain of the operator to account for a more
quantitative language, one is naturally lead to consider algebraic
$\lambda$-calculi, to which giving a clean computational meaning has
been proved hard so far~\cite{V09}.

But what about \emph{probabilistic} $\lambda$-calculi~\cite{JP89},
which have received quite some attention recently (see,
e.g.~\cite{EPT18,BDLGS16,VKS19}) due to their applicability to
randomised computation and bayesian programming? Can the Taylor
expansion naturally be generalised to those calculi?  This is an
interesting question, to which we give the first definite positive
answer in this paper. In particular, we show that the Taylor expansion
of probabilistic $\lambda$-terms is a conservative extension of the
well-known one on ordinary $\lambda$-terms.  In particular, the target
can be taken, as usual, as a linear combination of \emph{ordinary}
resource $\lambda$-terms, i.e., the same kind of structure which
Ehrhard and Regnier considered in their work on the Taylor expansion
of \emph{pure} $\lambda$-terms. We moreover show that the Taylor
expansion, as extended to probabilistic $\lambda$-terms, continues to
enjoy the nice properties it has in the deterministic realm.  In
particular, it is adequate as a way to give semantics to probabilistic
$\lambda$-terms, and the equational theory on probabilistic
$\lambda$-terms induced by Taylor expansion coincides with the one
induced by a probabilistic variation on B\"ohm trees~\cite{B84}. The
latter, noticeably, has been proved to capture observational
equivalence, one quotiented modulo $\eta$-equivalence~\cite{B84}.

Are we the first ones to embark on the challenge of generalising
Taylor's expansion to probabilistic $\lambda$-calculi, and in general
to effectful calculi?  Actually, some steps in this direction have
recently been taken. First of all, we need to mention the line of
works originated by Tsukada and Ong's paper on rigid resource
terms~\cite{TAO17}. This has been claimed from the very
beginning to be a way to model effects in the resource
$\lambda$-calculus, but it has also been applied to, among others,
probabilistic effects, giving rise to quantitative denotational
models~\cite{TAO18}.  The obtained models are based on species, and
are proved to be adequate. The construction being generic, there is no
aim at providing a precise comparison between the discriminating power of
the obtained theory and, say, observational equivalence: the choice of
the underlying effect can in principle have a huge impact on it.

One should also mention Vaux's work on the algebraic $\lambda$-calculus
\cite{V09}, where one can build arbitrary linear combinations of terms.
He showed a correspondence between Taylor expansion and B\"{o}hm trees,
but only for terms whose B\"{o}hm trees approximants at finite depths
are computable in a finite number of steps. This includes all ordinary
$\lambda$-terms but not all probabilistic ones. More recently Olimpieri
and Vaux have studied a Taylor expansion for a non-deterministic
$\lambda$-calculus \cite{OV18} corresponding to our notion of \emph{explicit}
Taylor expansion (Section~\ref{sec:taylor_rigid}).

In the rest of this section, probabilistic Taylor expansion will be
informally introduced by way of an example, so as to make the main
concepts comprehensible to the non-specialist.  In
sections~\ref{sec:rigid} and~\ref{sec:taylor_rigid}, we introduce a new
form of resource term, and a notion of explicit Taylor expansion from
probabilistic $\lambda$-terms. These constructions have an
interest in themselves (again, see \cite{OV18}) but in this paper they
are just an intermediate step towards proving our main
results. Definitionally, the crux of the paper is
Section~\ref{sec:taylor}, in which the Taylor expansion of a
probabilistic $\lambda$-term is made to produce \emph{ordinary}
resource terms. The relationship between the introduced theory and the
one induced by Probabilistic B\"ohm trees~\cite{L18} is investigated
in Section~\ref{sec:trees} and Section~\ref{sec:taylor_testing}.
\subsection*{The Probabilistic Taylor Expansion, Informally}
In this section, we introduce the main ingredients of the probabilistic
Taylor expansion by way of an extremely simple, although instructive,
example. Let us consider the probabilistic $\lambda$-term $M=\delta(I\oplus\Omega)$,
where $\oplus$ is an operator for binary, fair, probabilistic choice, 
$\delta=\lambda x.xx$, $I=\lambda.x.x$ and $\Omega=\delta\delta$ is 
a purely diverging, term. As such, $M$ is a term of a minimal, untyped,
probabilistic $\lambda$-calculus. Evaluation of $M$, if performed leftmost-outermost
is as in Figure~\ref{fig:examplereductiontree}. In particular, the probability of convergence for
$M$ is $\frac{1}{4}$.
\begin{wrapfigure}{R}{.3\textwidth}
  \fbox{
    \begin{minipage}{.28\textwidth}
      \centering
    \includegraphics[scale=0.9]{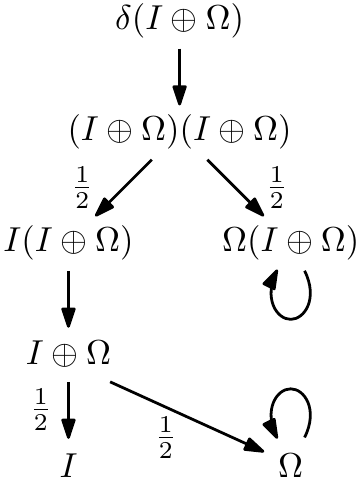}
  \end{minipage}}
  \caption{$M$'s Reduction Tree.}\label{fig:examplereductiontree}
\end{wrapfigure}
Please observe that two copies of the argument $I\oplus\Omega$ are produced, and
that the ``rightmost'' one is evaluated only when the ``leftmost'' one
converges, i.e. when the probabilistic choice $I\oplus\Omega$ produces $I$ as
a result.

The main idea behind building the Taylor expansion of any
$\lambda$-term $M$ is to describe the dynamics of $M$ by way of
\emph{linear approximations} of $M$. In the realm of the
$\lambda$-calculus, a linear approximation has traditionally been
taken as a \emph{resource $\lambda$-term}, which can be seen as a pure
$\lambda$-term in which applications have the form $\linapp{s}{\bag
  t}$, where $s$ is a term and $\bag t$ is a \emph{multiset} of terms,
and in which the result of firing the redex $\linapp{\lambda x.s}{\bag
  t}$ is the linear combination of all the terms obtained by
allocating the resources in $\bag t$ to the occurrences of $x$ in
$s$. For instance, one such element in the Taylor expansion of
$\Delta$ is $\lambda x.(\linapp{x}{[x]})$, where the occurrence of $x$
in head position is provided with only one copy of its argument. If
applied to the multiset $[y,z]$, this term would reduce into
$\linapp{y}{[z]} + \linapp{z}{[y]}$. Similarly, an element in the
Taylor expansion of $\Delta\ I$ would be $\linapp{\lambda
  x.\linapp{x}{[x]}}{[I^2]}$, which reduces into $2.\linapp{I}{[I]}$.
Another element of the same Taylor expansion is $\linapp{\lambda
  x.\linapp{x}{[x]}}{[I^3]}$, but this one reduces into $0$: there is
no way to use its resources linearly, i.e., using them without copying
and erasing. The actual Taylor expansion of a term is built by translating any
application $M\ N$ into an infinite sum $\tay{(M\ N)} = \sum_{n \in
  \Nat} \frac{1}{n!}.\linapp{\tay M}{[(\tay N)^n]}$.  For instance,
the Taylor expansion of $\Delta\ I$ is $\sum_{m,n \in \Nat}
\frac{1}{m!n!}.\linapp{\lambda x.\linapp{x}{[x^m]}}{[I^n]}$. Remark
that any summand properly reduces only when $n=m+1$, in which case
it reduces to $n!.\linapp{I}{[I^m]}$. In turn $\linapp{I}{[I^m]}$
reduces properly only when $m=1$, and the result is $I$. All the other
terms reduce to $0$. In the end the Taylor expansion of $\Delta\ I$
normalises to $\frac{2!}{1!2!}.I=I$.

Extending the Taylor expansion to probabilistic terms seems
straightforward, a natural candidate for the Taylor expansion of
$M\oplus N$ being just $\frac{1}{2}.\tay M + \frac{1}{2}.\tay N$. When
computing the Taylor expansion of $M$ we will find expressions such as
$\linapp{\lambda x.\linapp{x}{[x]}}{[(\frac{1}{2}.I + \frac{1}{2}.\tay
    \Omega)^2]}$, i.e.  $\frac{1}{4}.\linapp{\lambda
  x.\linapp{x}{[x]}}{[I^2]} + \frac{1}{4}.\linapp{\lambda
  x.\linapp{x}{[x]}}{[\Omega^2]} +\frac{1}{2}.\linapp{\lambda
  x.\linapp{x}{[x]}}{[I,\Omega]}$. For non-trivial reasons, the Taylor
expansion of any diverging term normalises to $0$, so just like in our
previous example, the only element in $\tay M$ which does not reduce
to $0$ is $\linapp{\lambda x.\linapp{x}{[x]}}{[I^2]}$. The difference
is that this time it appears with a coefficient
$\frac{1}{1!2!}\frac{1}{4}$, so $\tay M$ normalises to
$\frac{1}{4}.I$. Please notice how this is precisely the ``normal
form'' of the original term $M$. This is a general phenomenon, whose
deep consequences will be investigated in the rest of this paper, and
in particular in Section~\ref{sec:trees}.

\subsection*{Notations}
We write $\Nat$ for the set of natural numbers and $\Rpos$ for
the set of nonnegative real numbers.
Given a set $A$, we write $\Rpos^A$ for the set of families of positive
real numbers indexed by elements in $A$. We write such families as
linear combinations: an element $S \in \Rpos^A$ is a sum $S = \sum_{a
  \in A} S_a.a$, with $S_a \in \Rpos$. The support of a family $S \in
\Rpos^A$ is $\supp S = \{a \in A \mid S_a > 0\}$. We write $\Rpos^{(A)}$
for those families $S \in \Rpos^A$ such that $\supp S$ is finite.
Given $a \in A$ we often write $a$ for $1.a \in \Rpos^A$ unless we
want to emphasise the difference between the two expressions.
We also define \emph{finite
  multisets} over $A$ as functions $m : A \rightarrow
\Nat$ such that $m(a) \neq 0$ for finitely many $a \in A$. We use
the notation $[a_1,\dots,a_n]$ to describe the multiset $m$ such that
$m(a)$ is the number of indices $i \leq n$ such that $a_i = a$.

\section{Probabilistic Resource $\lambda$-Calculus}\label{sec:rigid}
In this section, we describe the theory of resource terms with
explicit choices, for the purpose of extending many of the properties
of resource terms to the probabilistic case. All this has an interest
in itself, but here this is mainly useful as a way to render certain
proofs about the Taylor Expansion easier (see
Section~\ref{sec:taylor_rigid} for more details). For this reason we
try to give the reader a clear understanding of this calculus and of
why these definitions and properties are useful, without focusing on
the actual proofs. These are straightforward generalisations of those
for deterministic resource terms \cite{ER08} and can be found in an
extended version of this paper \cite{EV}. The same results have
recently been given for a non-deterministic calculus \cite{OV18} by
Olimpieri and Vaux.
\subsection{The Basics}
\begin{definition}
  The sets of \Def{probabilistic simple resource terms} $\psrt$ and
  of \Def{probabilistic simple resource poly-terms} $\psrpt$ over a
  set of variables $\Var$ are defined by mutual induction as follows:
  $$
    s,t \in \psrt := x \mid \lambda x.s \mid \linapp{s}{\bag t} \mid \ls{p}{s} \mid \rs{p}{s}\qquad\qquad
    \bag s,\bag t \in \psrpt := [s_1,\dots,s_n]
  $$
  where $p$ ranges over $[0,1]$. We call \Def{finite probabilistic resource terms} the finite linear
  combinations of resource terms in $\fprt$, and \Def{finite
  probabilistic resource poly-terms} the finite linear combinations of
  resource poly-terms in $\fprpt$. We extend the constructors of
  simple (poly-)terms to (poly-)terms by linearity, e.g., if
  $S \in \fprt$ then $\lambda x.S$ is defined as the poly-term such
  that $(\lambda x.S)_{\lambda x.s} = S_s$ and $(\lambda x.S)_t = 0$
  if $t$ is not an abstraction.
\end{definition}

Some consecutive abstractions $\lambda x_1.\dots\lambda
x_n.s$ will be indicated as $\lambda x_1 \dots x_n.s$, or even as $\lambda \vec
x.s$. Similarly, to describe many successive applications
$\linapp{\linapp{\linapp{M}{N_1}}{\dots}}{N_k}$, we use a single pair
of brackets and we write $\linapp{M}{N_1\ \dots\ N_k}$.
%
We write $\pgsrt$ for $\psrt \cup \psrpt$, which is ranged over by
metavariables like $\sigma, \tau$. Note that intuitively $\pgsrt$
should stand for either $\psrt$ or $\psrpt$, not their union. For
instance we will prove some properties for finite linear combinations
in $\fpgrt$, but the only relevant linear combinations are the actual
(poly-)terms in $\fprt$ or $\fprpt$. Yet this distinction is
technically irrelevant, and all our results hold if we define $\pgsrt$
as a union.

The reason why linear combinations over such elements are
dubbed \emph{terms} will be clear once we describe the operational
semantics of the resource calculus. The main point of
the resource $\lambda$-calculus is to allow functions to use their
argument arbitrarily many times and yet remain entirely linear,
which is achieved by taking multisets as arguments: if a function uses
its argument $n$ times then it needs to receive $n$ resources as
argument and use each of them linearly. This idea has two
consequences. First, an application can fail if a function is not
given exactly as many arguments as it needs, as it would need either
to duplicate or to discard some of them. Second, the result of a
valid application is often not unique: a function can choose
how to allocate the different resources to the different calls to its
argument, and different choices may lead to different results. Both
these features are treated using linear combinations: a failed
application results in $0$ (i.e. the trivial linear combination) and a
successful one yields the sum of all its possible outcomes.

\begin{definition}
  We define the substitution of $\bag t \in \psrpt$ for $x \in \Var$
  in $\sigma \in \pgsrt$ by:
  \[\dsubst{\sigma}{[t_1,\dots,t_n]}{x} = \begin{cases} 0\ \text{if $\sigma$ does not have exactly $n$ free occurences of $x$}\\\sum_{\rho \in \mathfrak{S}_n} \sigma[t_{\rho(1)}/x_1,\dots,t_{\rho(n)}/x_n] \in \fpgrt\ \text{otherwise}\end{cases}\]
  where $x_1,\dots,x_n$ are the free occurrences of $x$ in $\sigma$
  and $\mathfrak{S}_n$ is the set of permutations over
  $\{1,\dots,n\}$. Alternatively, we could define
  $\dsubst{s}{\bag t}{x}$ by induction on $s$, as follows
  \begin{align*}
   \begin{aligned}
     \dsubst{x}{[t]}{x} &= t &&& \dsubst{(\lambda y.s)}{\bag t}{x} &= \lambda y.\dsubst{s}{\bag t}{x}\ \text{if}\ y \neq x\\
    \dsubst{y}{[]}{x} &= y\ \text{if}\ y \neq x &&& \dsubst{(\ls{p}{s})}{\bag t}{x} &= \ls{p}{\dsubst{s}{\bag t}{x}}\\
    \dsubst{z}{\bag t}{x} &= 0\ \text{in any other case} &&& \dsubst{(\rs{p}{s})}{\bag t}{x} &= \rs{p}{\dsubst{s}{\bag t}{x}}\\
  \end{aligned}\\
  \begin{aligned}
    \dsubst{(\linapp{s}{\bag u})}{[t_1,\dots,t_n]}{x} &= \sum_{I \uplus J = \{1,\dots,n\}} \linapp{\dsubst{s}{[t_i]_{i \in I}}{x}}{\dsubst{\bag u}{[t_j]_{j \in J}}{x}}\\
    \dsubst{[u_1,\dots,u_m]}{[t_1,\dots,t_n]}{x} &= \sum_{\biguplus_{k=1}^m I_k = \{1,\dots,n\}} [\dsubst{u_1}{[t_i]_{i \in I_1}}{x},\dots,\dsubst{u_m}{[t_i]_{i \in I_m}}{x}]
  \end{aligned}
  \end{align*}
  where $\uplus$ is the disjoint union of sets.
\end{definition}

\begin{example}
A basic example is $\dsubst{(\linapp{x}{[x]})}{[y,z]}{x}
= \linapp{y}{[z]} + \linapp{z}{[y]}$: there are two occurrences of $x$
in $\linapp{x}{[x]}$, so there are two ways to substitute $[y,z]$ for them. Remark that we
also have $\dsubst{[x,x]}{[y,z]}{x} = [y,z] + [z,y] = 2.[y,z]$: the
two occurrences of $x$ are not as clearly distinguished as in the
first example but they still count as different occurrences. Similarly
$\dsubst{(\linapp{x}{[x]})}{[y,y]}{x} = 2.\linapp{y}{[y]}$ and
$\dsubst{[x,x]}{[y,y]}{x} = 2.[y,y]$: there are two distinct
occurrences of $y$, so there are two ways to allocate them.  As another
example, please consider $\dsubst{(\lambda x.x)}{[y]}{x}
= \dsubst{(\linapp{x}{[x]})}{[y]}{x} = 0$: the substitution fails
if the number of resources does not match the number of free
occurrences of the substituted variable.
\end{example}

The operational semantics of the deterministic resource
$\lambda$-calculus~\cite{ER08} is usually given as a single rule of
$\beta$-reduction. In the probabilistic setting, we also need rules to
make choices commute with head contexts.

\begin{definition}
  The reductions $\rbeta$ and $\rsum$ are defined from $\pgsrt$ to $\fpgrt$ by:
  \begin{align*}
    \linapp{\lambda x.s}{\bag t} &\rbeta \dsubst{s}{\bag t}{x}\\
    \lambda x.(\ls{p}{s}) &\rsum \ls{p}{\lambda x.s} &&& \lambda x.(\rs{p}{s}) &\rsum \rs{p}{\lambda x.s}\\
    \linapp{\ls{p}{s}}{\bag t} &\rsum \ls{p}{\linapp{s}{\bag t}} &&& \linapp{\rs{p}{s}}{\bag t} &\rsum \rs{p}{\linapp{s}{\bag t}}
  \end{align*}
  extended under arbitrary contexts. We simply write $\rightarrow$ for $\rbeta \cup \rsum$.
  Reduction can be extended to finite terms in the following way: if
  $S \in \fpgrt$, $S_\sigma > 0$ and $\sigma \rightarrow T$ then
  $S \rightarrow S - S_\sigma.\sigma + S_\sigma T$.
\end{definition}

As the resource $\lambda$-calculus does not allow any duplication, and
$\beta$-reduction erases some constructors, it naturally decreases the
size of the involved simple terms. Consequently, $\beta$-reduction is
strongly normalising. This result can be extended to the whole
reduction $\rightarrow$, which is also confluent.

More specifically we define the \Def{size} $\size{\sigma}$ of a simple
(poly-)term in a natural way. To any $S \in \fpgrt$ we associate two
sizes: $\size S = 1 + \max_{\sigma \in \supp S} \size \sigma$ and
$\ssize S = [\size \sigma]_{\sigma \in \supp S} \in \Mf{\Nat}$. We
order $\Mf{\Nat}$ with a reverse lexicographical order: $m \prec n$
iff there exists $a \in \Nat$ such that $m(a) < n(a)$ and $m(b) =
n(b)$ for all $b > a$.

\begin{proposition}\label{prop:finite_normalisation}
  The reduction $\rightarrow$ is confluent and strongly normalising on
  $\fpgrt$. Given $S \in \fpgrt$ we write $\nf S$ for its unique
  normal form for $\rightarrow$, and given $\sigma \in \pgsrt$ we
  write $\nf \sigma$ for $\nf{1.\sigma}$.
\end{proposition}
\begin{proof}
  Proving weak confluence is straightforward. Strong
  normalisation is proven in two steps. First using an appropriate
  weight on terms describing how deep choices are we can prove that
  $\rsum$ is strongly normalising. Second one can observe that $\rsum$
  preserves size, and that if $\sigma \rbeta T$ and $\tau \in \supp T$
  then $\size \tau < \size \sigma$, hence if $S \rbeta T$ then $\ssize
  T \prec \ssize S$. The confluence is given by Newman's Lemma.
\end{proof}

\subsection{Complete Left Reduction}

This reduction is not convenient to study (poly-)terms with particular
properties such as \emph{uniformity} or \emph{regularity}, which we
will define later. For instance given a simple poly-term $\bag s =
[s,\dots,s]$ we can reduce independently the different occurrences of
$s$, so not every reduct of $\bag s$ is of the form $[T,\dots,T]$ with
$s \twoheadrightarrow T$. Similarly given a term $S$ we can reduce
independently the elements of its support, possibly losing some common
properties shared by these elements. For that reason (as well as the
issue of infinite terms discussed in the rest of this section) we are
mostly interested in normalisation rather than reduction. To study
this normalisation we still need some small-step operational
semantics, but it will be more convenient to consider
the \emph{complete left reduction} defined as follows.

\begin{definition}
  We define the \Def{complete left reduct} $\Lred \sigma \in \fpgrt$ of a simple (poly-)term $\sigma$ by induction:
  \begin{align*}
    \Lred{\ls{p}{s}} &= \ls{p}{\Lred{s}}\\
    \Lred{\rs{p}{s}} &= \rs{p}{\Lred{s}}\\
    \Lred{\lambda \vec x.\linapp{y}{\bag u_1\,\dots\,\bag u_m}} &= \lambda \vec x.\linapp{y}{\Lred{\bag u_1}\,\dots\,\Lred{\bag u_m}}\\
    \Lred{\lambda \vec x.\linapp{\lambda y.s}{\bag t\,\bag u_1\,\dots\,\bag u_m}} &= \lambda \vec x.\linapp{\dsubst{s}{t}{y}}{\bag u_1\,\dots\,\bag u_m}\\
    \Lred{\lambda \vec x.\linapp{\ls{p}{s}}{\bag u_1\,\dots\,\bag u_m}} &= \ls{p}{(\lambda \vec x.\linapp{s}{\bag u_1\,\dots\,\bag u_m})}\\
    \Lred{\lambda \vec x.\linapp{\rs{p}{s}}{\bag u_1\,\dots\,\bag u_m}} &= \rs{p}{(\lambda \vec x.\linapp{s}{\bag u_1\,\dots\,\bag u_m})}\\
    \Lred{[s_1,\dots,s_n]} &= [\Lred{s_1},\dots,\Lred{s_n}]
  \end{align*}
  We extend this definition to terms: $\Lred S = \sum_{\sigma \in \pgsrt} S_\sigma\Lred \sigma$.
\end{definition}

\begin{proposition}\label{prop:nf_Lred}
  For all $S \in \fpgrt$, $S \twoheadrightarrow \Lred{S}$.
\end{proposition}

\begin{proposition}\label{prop:nf_Lred}
  For all $S \in \fpgrt$ there is $k \in \Nat$ such that $\nf S = \Lred[k]{S}$.
\end{proposition}

\begin{proof}
  The reduction $\rightarrow$ being strongly normalising we reason by induction on the bound on the length of the reductions of $S$. We have either $\Lred S = S$ and $S$ is already in normal form or $S$ reduces into $\Lred S$ in a least one step and we conclude by induction hypothesis.
\end{proof}

\subsection{Infinite Terms}
So far we only worked with finite terms but to fully express the
operational behaviour of a $\lambda$-term in the resource
$\lambda$-calculus, which is the purpose of the Taylor expansion, we need
infinite ones. We can extend the constructors of the
calculus to $\pgrt$ by linearity and generalise the reduction relation
$\rightarrow$, but Proposition~\ref{prop:finite_normalisation}
fails. Indeed let $I_0 = I = \lambda x.x$ and $I_{n+1}
= \linapp{I_n}{[I]}$. For $n \in \Nat$, let $S = \sum_{n \in \Nat}
I_n$. Then, for all $n \in \Nat$ the term $I_n$ normalises in $n$ steps and
$S$ \emph{does not} normalise in a finite number of reduction steps. A simple
solution to this problem is to define the ``normal form'' of an
infinite term by normalising each of its components: we can set $\nf S
= \sum_{\sigma \in \pgsrt} S_\sigma\nf \sigma$. But then another problem
arises. In our previous example, we have $\nf{I_n} = I$ for all
$n \in \Nat$, thus we would have $\nf S = \sum_{n \in \Nat} I$, which
is not an element of $\pgrt$ as the coefficient of $I$ is infinite.
Still we can use this pointwise normalisation if we consider terms
with a particular property, called \emph{uniformity}.

\begin{definition}
  The \Def{coherence relation} $\coh$ on $\pgsrt$ is defined by:
  $$
  \AxiomC{$\phantom{s \coh s'}$} \UnaryInfC{$x \coh x$} \DisplayProof \qquad
  \AxiomC{$s \coh s'$} \UnaryInfC{$\lambda x.s \coh \lambda x.s'$} \DisplayProof \qquad
  \AxiomC{$s \coh s'$} \AxiomC{$\bag t \coh \bag{t'}$} \BinaryInfC{$\linapp{s}{\bag t} \coh \linapp{s'}{\bag{t'}}$} \DisplayProof \qquad
  \AxiomC{$s \coh s'$} \UnaryInfC{$\ls{p}{s} \coh \ls{p}{s'}$} \DisplayProof
  $$
  $$
  \AxiomC{$s \coh s'$} \UnaryInfC{$\rs{p}{s} \coh \rs{p}{s'}$} \DisplayProof \qquad
  \AxiomC{$s \coh s$} \AxiomC{$t \coh t$} \BinaryInfC{$\ls{p}{s} \coh \rs{p}{t}$} \DisplayProof \qquad
  \AxiomC{$\forall i,j \leq m+n, s_i \coh s_j$} \UnaryInfC{$[s_1,\dots,s_m] \coh [s_{m+1},\dots,s_{m+n}]$} \DisplayProof
  $$
  For $S,S' \in \pgsrt$ we write $S \coh S'$ when for all
  $\sigma,\sigma' \in \supp S \cup \supp{S'}$, $\sigma \coh \sigma'$.
  A simple (poly-)term $\sigma \in \pgsrt$ is called \Def{uniform} if
  $\sigma \coh \sigma$, and a term $S \in \pgrt$ is
  called \Def{uniform} if $S \coh S$.
\end{definition}

\begin{remark}
In the rule for $\ls{p}{s} \coh \rs{p}{t}$ we require $s \coh s$ and
$t \coh t$ to ensure that whenever $\sigma \coh \tau$, the simple
(poly-)terms $\sigma$ and $\tau$ are necessarily uniform. This is not
crucial, as we will only consider uniform (poly-)terms, whose support
contains only uniform simple (poly-)terms by definition, but this simplifies
inductive reasoning.
\end{remark}


What makes coherence and uniformity interesting is that if two coherent
terms $S$ and $S'$ have disjoint supports, then all of their reducts, and in
particular their normal forms, have disjoint supports. Then any
element in the support of $\nf{S + S'}$ comes either from $\nf S$ or
from $\nf{S'}$, but it cannot come from both.
\begin{lemma}\label{lem:subst}
  If $\sigma \coh \sigma'$ and $\bag u \coh \bag{u'}$ then
  $\dsubst{\sigma}{\bag
  u}{x} \coh \dsubst{\sigma'}{\bag{u'}}{x}$. Besides if
  $\supp{\dsubst{\sigma}{\bag
  u}{x}} \cap \supp{\dsubst{\sigma'}{\bag{u'}}{x}} \neq \emptyset$
  then $\sigma = \sigma'$ and $\bag u = \bag{u'}$.
\end{lemma}
\begin{proof}
By induction on $\sigma \coh \sigma'$:
\begin{varitemize}
\item
If $x \coh x$ then for $\supp{\dsubst{x}{\bag u}{x}}$ and
$\supp{\dsubst{x}{\bag {u'}}{x}}$ to be both nonempty we need to have
$\bag u = [v]$ and $\bag{u'} = [v']$ for some $v,v' \in \Delta^+$, and
in this case $\dsubst{x}{\bag u}{x} = v$ and $\dsubst{x}{\bag{u'}}{x}
= v'$. The hypothesis $\bag u \coh \bag{u'}$ implies $v \coh v'$, and
if $v = v'$ then $\bag u = \bag{u'}$.

\item
If $y \coh y$ with $y \neq x$ then either one of the substitutions is
$0$ or we have $u = u' = [\ ]$.

\item
If $\lambda x.s \coh \lambda x.s'$, $\ls{p}{s} \coh \ls{p}{s'}$ or
$\rs{p}{s} \coh \rs{p}{s'}$, with in each case $s \coh s'$, then the
result is immediate by induction hypothesis.

\item
If $\ls{p}{s} \coh \rs{p}{t}$ then we use the induction hypothesis on
$s \coh s$ and $u \coh u$ (given by Proposition \ref{prop:coh_ref}) to
prove that for $v \in \supp{\dsubst{s}{\bag u}{x}}$ we have $v \coh
v$, and similarly for $w \in \supp{\dsubst{\bag t}{\bag{u'}}{x}}$, and
the result follows. Notice that we will never have $\ls{p}{v} =
\rs{p}{w}$.

\item
If $\linapp{s}{\bag t} \coh \linapp{s'}{\bag{t'}}$ then
$\supp{\dsubst{\linapp{s}{\bag t}}{\bag u}{x}} = \bigcup_{I \uplus J =
  [1,\#\bag u]} \{\linapp{v}{\bag w} \; v \in \supp{\dsubst{s}{\bag
    u_I}{x}}, \bag w \in \supp{\dsubst{\bag t}{\bag u}{x}}\}$, and
similarly for $\linapp{s'}{\bag{t'}}$. Observe that for $I \uplus J =
[1,\#\bag u]$ and $I' \uplus J' = [1,\#\bag {u'}]$ we have $u_I \coh
u'_{I'}$ and $u_J \coh u'_{J'}$ so we can apply the induction
hypothesis to $s \coh s'$ and $u_I \coh u'_{I'}$ and to $\bag t \coh
\bag{t'}$ and $u'_J \coh u_{J'}$ to get the result.

\item
Finally if $\bag s = [s_1,\dots,s_m] \coh [s_{m+1},\dots,s_{m+n}] =
\bag{s'}$ we use a similar reasoning: for any $I,I' \subset [1,\#\bag
  u]$ and $J,J' \subset [1,\#\bag {u'}]$ we have $u_I \coh u_{I'}$,
$u'_J \coh u'_{J'}$ and $u_I \coh u'_J$, hence by induction hypothesis
for any $v,v' \in \bigcup_{i \leq m}\bigcup_{I \subset [1,\#\bag u]}
\supp{\dsubst{s_i}{\bag u_I}{x}}$ and $w,w' \in \bigcup_{j \leq
  n}\bigcup_{J \subset [1,\#\bag {u'}]} \supp{\dsubst{s_{m+j}}{\bag
    {u'}_J}{x}}$ we have $v \coh v'$, $w \coh w'$ and $v \coh w$. This
gives the first part of the result. Now if $\bag v = [v_1,\dots,v_k]
\in \supp{\dsubst{\bag s}{\bag u}{x}} \cap
\supp{\dsubst{\bag{s'}}{\bag{u'}}{x}}$ then necessarily $n = m = k$,
and we can find sets $I_i$ and $J_i$ such that $\biguplus_{i \leq k}
I_i = [1,\#\bag u]$, $\biguplus_{i \leq k} J_i = [1,\#\bag {u'}]$ and
$v_i \in \supp{\dsubst{s_i}{\bag u_{I_i}}{x}} \cap
\supp{\dsubst{s_{k+i}}{\bag {u'}_{J_i}}{x}}$ (up to permutation of the
indices in $\bag s$ and $\bag{s'}$). By induction hypothesis we get
$s_i = s_{k+i}$ and $\bag u_{I_i} = \bag{u'}_{J_i}$, hence $\bag s =
\bag{s'}$ and $\bag u = \bag{u'}$.

\end{varitemize}
\end{proof}

\begin{proposition}\label{prop:Lred_uniform_distinct}
  Given $S,S' \in \fpgrt$, if $S \coh S'$ then $\Lred S \coh \Lred{S'}$. If moreover $\supp S \cap \supp{S'}
  = \emptyset$ then $\supp{\Lred S} \cap \supp{\Lred{S'}} = \emptyset$.
\end{proposition}

\begin{proof}
It is sufficient to prove the result for simple terms $\sigma,\sigma'$ as the generalisation to finite terms is straightforward. We reason by induction on $\sigma$ and the proof of $\sigma \coh \sigma'$.
\begin{itemize}
  \item If $\ls{p}{s} \coh \ls{p}{s'}$ or $\rs{p}{s} \coh \rs{p}{s'}$ the result is immediate by induction hypothesis.
  \item If $\ls{p}{s} \coh \rs{p}{s'}$ then $s$ and $s'$ are uniform and by induction hypothesis so are $\Lred s$ and $\Lred{s'}$, hence $\ls{p}{\Lred s} \coh \rs{p}{\Lred{s'}}$.
  \item The case of head normal forms is immediate by induction hypothesis.
  \item If $\lambda \vec x.\linapp{\lambda y.s}{\bag t\,\bag u_1\,\dots\,\bag u_m} \coh \lambda \vec x.\linapp{\lambda y.s'}{\bag t'\,\bag u'_1\,\dots\,\bag u'_m}$ then we apply Lemma~\ref{lem:subst}.
  \item The cases of head choices are immediate.
  \item The case of poly-terms is immediate by induction hypothesis.
\end{itemize}
\end{proof}

\begin{corollary}\label{prop:nf_uniform_distinct}
  Given $S,S' \in \fpgrt$, if $S \coh S'$ then $\nf S \coh \nf{S'}$. If moreover $\supp S \cap \supp{S'}
  = \emptyset$ then $\supp{\nf S} \cap \supp{\nf{S'}} = \emptyset$.
\end{corollary}
\begin{proof}
Using Proposition~\ref{prop:nf_Lred}, by induction on $k$.
\end{proof}

This immediately implies that pointwise reduction of infinite uniform terms
is well defined, as both complete left reducts and normal forms of distinct but coherent simple (poly-)terms
have disjoint supports.

\begin{corollary}
  If $S \in \pgrt$ is uniform then $\sum_{\sigma \in \pgsrt}
  S_s\Lred \sigma$ and $\sum_{\sigma \in \pgsrt}
  S_\sigma\nf \sigma$ are in $\pgrt$. We write $\Lred S$ and $\nf S$ respectively for these sums.
\end{corollary}
\begin{proof}
For all $\sigma \neq \sigma' \in \supp S$ we have by hypothesis
$\sigma \coh \sigma'$ so the previous proposition gives
$\supp{\Lred \sigma} \cap \supp{\Lred{\sigma'}} = \emptyset$. Therefore
given any $\tau \in \pgsrt$ there is at most one $\sigma \in \supp S$
such that $\tau \in \supp{\Lred \sigma}$. The same goes for normalisation.
\end{proof}

\begin{remark}
Although both complete left reduction and normal forms are well defined for infinite terms, Proposition~\ref{prop:nf_Lred} doesn't hold: consider $\bag s_0 = [\ ]$, $\bag s_{n+1} = [\linapp{\lambda x.x}{\bag s_n}]$ and $\bag S = \sum_{n \in \Nat} \bag s_n$, then $\bag S$ is uniform and $\nf{\bag S} = 0$ but for all $k \in \Nat$, $\Lred[k]{\bag S} = \bag S \neq 0$. Besides $\nf{\bag S}$ is not even the limit of the $\Lred[k]{\bag S}$ as $k$ approaches $\infty$. However normal forms are indeed limits of complete left reducts \emph{restricted to normal simple terms}.
\end{remark}

\begin{proposition}\label{prop:nf_limit_Lred}
  Given a uniform (poly-)term $S \in \pgrt$ and given $\tau \pgsrt$ in
  normal form, we have $\nf S_\tau = \Lred[k] S_\tau$ for all
  $k \in \Nat$ large enough.
\end{proposition}

\begin{proof}
  If $\tau \in \supp{\nf S}$ then by
  Corollary~\ref{prop:nf_uniform_distinct} there is a unique
  $\sigma \in \supp S$ such that $\tau \in \supp{\nf \sigma}$, and by
  Proposition~\ref{prop:nf_Lred} for all $k \in \Nat$ large enough we
  have $\nf \sigma_\tau = \Lred[k] \sigma_\tau$.
\end{proof}

\subsection{Regular Terms}
The deterministic Taylor expansion associates to any $\lambda$-term a
\emph{uniform} term, and explicit choices are adopted precisely for the sake
of preserving this property in the probabilistic case. Taylor
expansions have another important property: they are entirely defined
by their support. If a simple term $s$ is in the support of the Taylor
expansion of a $\lambda$-term $M$, then its coefficient is the inverse
of its \emph{multinomial coefficient}, which does not depend on
$M$. Moreover this property is preserved by normalisation. Using
explicit choices enforces this result in the probabilistic case, as
well.

\begin{definition}
  For any $\sigma \in \pgsrt$ we define the \Def{multinomial coefficient} $\mcoef \sigma\in\Nat$ by:
  \begin{align*}
    \mcoef x &= 1 &&& \mcoef{\linapp{s}{\bag t}} &= \mcoef s\mcoef{\bag t}\\
    \mcoef{\lambda x.s} &= \mcoef{\ls{p}{s}} = \mcoef{\rs{p}{s}} = \mcoef s &&&
    \mcoef{\bag s} &= \prod_{u \in \psrt} \bag s(u)!\cdot\mcoef u^{\bag s(u)}
  \end{align*}
  where $\bag s(u)$ is the multiplicity of $u$ in $\bag s$.
\end{definition}

\begin{definition}
  A uniform term $S \in \pgrt$ is called \Def{regular} if for all
  $\sigma \in \supp S$, $S_\sigma = \frac{1}{\mcoef \sigma}$.
\end{definition}
Multinomial coefficients correspond to the number of permutations of
multisets which preserve the description of simple (poly-)terms. For
instance, given variables $x_1,\dots,x_n \in \Var$, the coefficient
$\mcoef{[x_1,\dots,x_n]}$ is exactly the number of permutations $\rho \in \mathfrak S_n$
such that $(x_{\rho(1)},\dots,x_{\rho(n)}) =
(x_1,\dots,x_n)$. For a more precise interpretation of multinomial
coefficients see \cite{ER08} or \cite{TAO17}. Due to their relation with 
permutations in multisets, these coefficients appear naturally when
we perform substitutions.

\begin{theorem}\label{thm:substitution_multinomial}
  For any $\sigma \in \pgsrt$ uniform, for $x \in \Var$, $\bag
  t \in \psrpt$ and $u \in \supp{\dsubst{\sigma}{\bag t}{x}}$, we
  have: $(\dsubst{\sigma}{\bag t}{x})_u = \frac{\mcoef{\bag
  t}\mcoef \sigma}{\mcoef u}$.
\end{theorem}
There exist two methods to prove similar theorems in the literature,
and both can be used to prove
Theorem~\ref{thm:substitution_multinomial}. The first one is the
original proof by Ehrhard and Regnier for the pure deterministic
case \cite{ER08}, and its generalisation is straightforward and only
requires to extend the notion of uniformity (to take into account
that $[\ls{p}{s},\rs{p}{t}]$ is uniform). The second one is by Asada,
Tsukada and Ong for a simply typed calculus with choices \cite{TAO17}, and
it has been extended to the untyped case by Olimpieri and Vaux in an
unpublished paper \cite{OV18}.
We present here a direct generalisation of the proof in \cite{ER08}.

\begin{definition}
  A \Def{multilinear-free (poly)-term} is a (poly)-term $\varphi \in
  (!)\Delta^+$ such that all of its variables are free and each one
  occurs exactly once.  A \Def{multilinear-free substitution} is a
  partial function $\Phi$ from $\Var$ to multilinear-free terms such
  that $\Var(\Phi(x)) \cap \Var(\Phi(x')) = \emptyset$ for all $x \neq
  x'$ in $\Dom \Phi$.  We say that $(\varphi,\Phi)$ is \Def{adapted}
  if $\Var(\varphi) \subset \Dom \Phi$ and no element of $\Var(\Phi)$
  is bound in $\varphi$. Then $\Phi\varphi$ is the multilinear-free
  (poly)-term obtained by applying $\Phi$ on the variables of
  $\varphi$.  Similarly for any multilinear-free (poly)-term $\varphi$
  and $p : \Var(\varphi) \rightarrow \Var$ we write $p\varphi$ for the
  term obtained by applying $p$ to the variables of
  $\varphi$ \emph{without renaming captured variables}.  A pair
  $(\varphi,p)$ is said to \Def{represent} $\sigma \in \Delta^+$ if
  $p\varphi = \sigma$.
\end{definition}

\begin{definition}
  We define the following sets of bijections over variables:
  \begin{align*}
    \Sigma_p &= \{f : \Dom p \rightarrow \Dom p\ \text{bijective} \sep pf = p\}\\
    \Iso{\varphi,p} &= \{f \in \Sigma_p \sep f\varphi = \varphi\}\\
    \Iso{p,\Phi,q} &= \{g \in \Sigma_q \sep \exists f \in \Sigma_p : g\Phi = \Phi f\}\\
    \Iso{\varphi,p,\Phi,q} &= \{f \in \Sigma_p \sep \exists g \in \Sigma_q : g\Phi\varphi = \Phi f\varphi\}
  \end{align*}
\end{definition}

\begin{lemma}
  $|\Iso{\varphi,p}| = \mcoef{p\varphi}$.
\end{lemma}

\begin{lemma}
  For any $g \in \Iso{p,\Phi,q}$ there exists a unique $\pi(g) \in \Sigma_q$ such that $g\Phi = \Phi\pi(g)$, and $\pi : \Iso{p,\Phi,q} \rightarrow \Sigma_p$ is a group homomorphism.
\end{lemma}

\begin{lemma}
  $\pi(\Iso{p,\Phi,q})\Iso{\varphi,p} \subset \Iso{\varphi,p,\Phi,q}$.
\end{lemma}

\begin{definition}
  We define by induction a notion of \Def{uniformity} for pairs $(F,p)$ where $F$ is a multilinear-free polyterm and $p : \Var(F) \rightarrow \Var$:
  \begin{varitemize}
    \item $([x_1,\dots,x_n],p)$ is uniform if $p(x_i) = p(x_j)$ for all $i,j$;
    \item $([\lambda x.\varphi_1,\dots,\lambda x.\varphi_n],p)$ is uniform if $([\varphi_1,\dots,\varphi_n],p)$ is uniform;
    \item $([\linapp{\varphi_1}{G_1},\dots,\linapp{\varphi_n}{G_n}],p)$ is uniform if $([\varphi_1,\dots,\varphi_n],q)$ and $(G_1+\dots+G_n,r)$ are uniform, with $q$ and $r$ the obvious restrictions of $p$;
    \item $([\ls{p}{\varphi_1},\dots,\ls{p}{\varphi_n},\rs{p}{\varphi'_1},\dots,\rs{p}{\varphi'_{n'}}],p)$ is uniform if $([\varphi_1,\dots,\varphi_n],q)$ and $([\varphi'_1,\dots,\varphi'_{n'}],q')$ are uniform, where $q$ and $q'$ are the obvious restrictions of $p$.
  \end{varitemize}
  If $\varphi$ is a multilinear-free simple term we say that $(\varphi,p)$ is uniform if $([\varphi],p)$ is uniform.
\end{definition}

\begin{lemma}
  A pair $(\varphi,p)$ is uniform iff $p\varphi$ is uniform (i.e. $p\varphi \coh p\varphi$).
\end{lemma}

\begin{lemma}
  For $(\varphi,p)$ a uniform pair and $\Phi,\Phi'$ two multilinear-free substitutions over $\Var(\varphi)$, if $\Phi\varphi = \Phi'\varphi$ then there exists $f \in \Iso{\varphi,p}$ such that $\Phi'=\Phi f$.
\end{lemma}

\begin{lemma}
  If $(\varphi,p)$ is uniform then $\Iso{\varphi,p,\Phi,q} \subset \pi(\Iso{p,\Phi,q})\Iso{\varphi,p}$.
\end{lemma}

\begin{proposition}
  If $(\varphi,p)$ is uniform then $|\Iso{\varphi,p,\Phi,q}| = \frac{|\Iso{p,\Phi,q}||\Iso{\varphi,p}|}{|\Iso{\Phi\varphi,q}|}$
\end{proposition}

\begin{proof}
We have $|\pi(\Iso{p,\Phi,q})\Iso{\varphi,p}| = \frac{|\pi(\Iso{p,\Phi,q})||\Iso{\varphi,p}|}{|\pi(\Iso{p,\Phi,q}) \cap \Iso{\varphi,p}|}$.

Observe that $|\pi(\Iso{p,\Phi,q})| = \frac{|\Iso{p,\Phi,q}|}{\ker \pi}$ and $|\pi(\Iso{p,\Phi,q}) \cap \Iso{\varphi,p}| = |\ker \pi||\Iso{\Phi\varphi,q}|$.
\end{proof}

This is enough to conclude the proof of Theorem~\ref{thm:substitution_multinomial}.

This theorem ensures that a regular $\beta$-redex
$\frac{1}{\mcoef{\linapp{\lambda x.s}{\bag t}}}.\linapp{\lambda
x.s}{\bag t}$ reduces into a regular term. More generally, the theorem
is the key step towards proving that regular (poly-)terms always
normalise to regular (poly-)terms.

\begin{proposition}
  If $\sigma$ is uniform then for any $\tau \in \supp{\Lred \sigma}$, $\Lred \sigma_\tau = \frac{\mcoef \sigma}{\mcoef \tau}$.
\end{proposition}

\begin{proof}
  We reason by induction on $\sigma$, using Theorem~\ref{thm:substitution_multinomial} when dealing with $\beta$-reduction. Observe that in the case of a poly-term $\bag s = [s_1,\dots,s_n]$, according to Proposition~\ref{prop:Lred_uniform_distinct} for all $i,j \leq n$ we have either $s_i = s_j$ or $\supp{\Lred{s_i}} \cap \supp{\Lred{s_j}} = \emptyset$. This means that for a poly-term $\bag t = [t_1,\dots,t_n] \in \supp{\Lred{\bag s}}$ the number of pairwise distinct sequences $(t_{\rho(1)},\dots,t_{\rho(n)})$ with $\rho \in \mathfrak S_n$ such that $t_{\rho(i)} \in \supp{\Lred{s_i}}$ for all $i \leq n$ is exactly $\frac{\prod_{u \in \psrt} \bag s(u)!}{\prod_{v \in \psrt} \Lred{\bag s}(v)!}$.
\end{proof}

\begin{corollary}
  For all finite regular term $S$, $\Lred S$ and $\nf S$ are regular.
\end{corollary}

\begin{theorem}\label{thm:nf_regular}
  If $S \in \pgrt$ is regular then $\nf S$ is regular.
\end{theorem}
\begin{proof}
This follows directly from the previous result and Corollary~\ref{prop:nf_uniform_distinct}.
\end{proof}

\subsection{Regularity and the Exponential}
The regularity of terms is preserved by the constructors of simple
resource terms.

\begin{proposition}\label{prop:context_regular}
  For all $x \in \Var$, $S \in \prt$ regular and $\bag T \in \prpt$
  regular, the terms $1.x$, $\lambda x.S$, $\ls{p}{S}$, $\rs{p}{S}$ and $\linapp{S}{\bag T}$ are
  regular.
\end{proposition}

One may expect a similar result for poly-terms: if $S_1$,\dots,$S_n$
in $\prt$ are regular then $[S_1,\dots,S_n]$ is regular. However, this
is not the case: $1.x$ is regular and yet $1.[x,x]$ is not. Indeed
nontrivial coefficients appear in $\mcoef \sigma$ precisely when
$\sigma$ contains simple poly-terms with multiplicities greater than
$1$, so the regular sum with the same support as $[S_1,\dots,S_n]$ has
no simple description. A natural way to build regular poly-terms from
regular terms is to use the following construction.

\begin{definition}
  The \Def{exponential} of $S \in \prt$ is $!S
  = \sum_{n \in \Nat} \frac{1}{n!} [S^n] \in \prpt$, where $[S^n]$
  stands for the poly-term $[S,\dots,S]$ with $n$ copies of $S$.
\end{definition}

\begin{proposition}\label{prop:exp_regular}
  If $S \in \prt$ is regular then $!S$ is regular.
\end{proposition}
\begin{proof}
The key point is that the number of sequences $(s_1,\dots,s_n)$ which
describe a given simple poly-term $\bag s = [s_1,\dots,s_n]$ is
exactly $\frac{n!}{\prod_{u \in \psrt}\bag s(u)!}$.
\end{proof}

With these results, we have all the ingredients we need to translate
(probabilistic) $\lambda$-terms into regular terms: variables and abstractions of
regular terms are regular, and we can define an application between
regular terms following Girard's call-by-name translation of
intuitionistic logic into linear logic~\cite{G87}: $S$ applied to $T$ is
$\linapp{S}{!T}$.

\section{Explicit Probabilistic Taylor Expansion}\label{sec:taylor_rigid}
This section is devoted to defining and studying the \emph{Taylor expansion
with explicit choices}, or \emph{explicit Taylor expansion}, of
probabilistic $\lambda$-terms. It is named as such because its target is
the set of probabilistic resource terms, as defined in the previous section,
rather than the usual ones. This is \emph{not} the main contribution of this
paper, but an intermediate step in the study of Taylor
expansion as defined in Section~\ref{sec:taylor}.
\subsection{The Definition}
Probabilistic $\lambda$-terms are $\lambda$-terms enriched with a probabilistic
choice operator.
\begin{definition}
  The set of \Def{probabilistic $\lambda$-terms} $\Lambda^+$ is:
  \[M,N \in \Lambda^+ := x \mid \lambda x.M \mid M\ N \mid M \oplus_p N\]
\end{definition}
\begin{example}
  Let us consider the probabilistic $\lambda$-term $Q=\Delta(I +_{\frac{1}{2}}\Omega)$,
  where $\Delta=\lambda x.xx$, $I=\lambda x.x$, and $\Omega$ is any diverging
  term, e.g. $\Delta\Delta$. The term converges (to $I$) with probability
  $\frac{1}{4}$, and will be used as a running example throughout this section.
\end{example}

\begin{definition}
  The \Def{explicit Taylor expansion} $\rtay M$ is defined inductively as follows:
  \begin{align*}
    \rtay x &= x &&& \rtay{(M\ N)} &= \linapp{\rtay M}{!\rtay N} = \sum_{n \in \Nat} \frac{1}{n!}\linapp{\rtay M}{[(\rtay N)^n]}\\
    \rtay{(\lambda x.M)} &= \lambda x.\rtay M &&& \rtay{(M \oplus_p N)} &= (\ls{p}{\rtay M}) + (\rs{p}{\rtay N})
  \end{align*}
\end{definition}

\begin{definition}
  The support $\rTsupp{M} \subset \Delta^+$ of the Taylor expansion of $M \in \Lambda^+$ is defined by:
  \begin{align*}
    \rTsupp{x} &= \{x\}\\
    \rTsupp{\lambda x.M} &= \{\lambda x.s \sep s \in \rTsupp{M}\}\\
    \rTsupp{M\ N} &:= \{\linapp{s}{\bag t} \sep s \in \rTsupp{M}, \bag t \in \Mf{\rTsupp{N}}\}\\
    \rTsupp{M +_p N} &:= \{\ls{p}{s} \sep s \in \rTsupp{M}\} \cup \{\rs{p}{t} \sep t \in \rTsupp{N}\}
  \end{align*}
\end{definition}
\begin{proposition}
  For every $M \in \Lambda^+$, it holds that
  \[\rtay M = \sum_{s \in \rTsupp M} \frac{1}{\mcoef s} s \in \Rpos^{\Delta^+}.\]
\end{proposition}
\begin{proof}
  By induction on the structure of $M$:
  \begin{varitemize}
  \item
    If $M$ is a variable $x$, then
    $$
    \rtay{x}=\sum_{s \in \rTsupp x} \frac{1}{\mcoef s} s=\frac{1}{\mcoef x}x=x.
    $$
  \item
    If $M$ is an abstraction $\lambda x.N$, then:
    $$
      \rtay{(\lambda x.N)}=\sum_{s \in \rTsupp{N}} \frac{1}{\mcoef{\lambda x.s}}(\lambda x.s)
         =\lambda x.\left(\sum_{s \in \rTsupp{N}} \frac{1}{\mcoef{s}}s\right)=\lambda x.\rtay{N}\\
    $$
  \item
    If $M$ is an application $NL$, then we can first of
    all give the following lemma. For every $\bag t\in\Mfp{X}{n}$, it holds that
    $$
    |\{(t_1,\ldots,t_n)\sep[t_1,\ldots,t_n]=\bag t\}|=\frac{n!}{\prod_u\bag{t}(u)!}
    $$
    As a consequence,
    \begin{align*}
      \rtay{(NL)}&=\sum_{s \in \rTsupp{NL}} \frac{1}{\mcoef s} s = \sum_{s\in\rTsupp{N}} \sum_{\bag t\in\Mf{\rTsupp{L}}}\frac{1}{\mcoef{\linapp{s}{\bag t}}} \linapp{s}{\bag t}\\
         &=\sum_{s\in\rTsupp{N}}\sum_{n\in\Nat} \sum_{\bag t\in\Mfp{\rTsupp{L}}{n}}\frac{1}{\mcoef{\linapp{s}{\bag t}}} \linapp{s}{\bag t}\\
         &=\sum_{s\in\rTsupp{N}}\sum_{n\in\Nat} \sum_{\bag t\in\Mfp{\rTsupp{L}}{n}}\frac{1}{\mcoef{s}\mcoef{\bag t}} \linapp{s}{\bag t}\\
         &=\sum_{s\in\rTsupp{N}}\sum_{n\in\Nat} \sum_{t_1,\ldots,t_n\in\rTsupp{L}}\frac{\prod_u[t_1,\ldots,t_n](u)}{\mcoef{s}\cdot\mcoef{[t_1,\ldots,t_n]}\cdot n!} \linapp{s}{[t_1,\ldots,t_n]}\\
         &=\sum_{s\in\rTsupp{N}}\sum_{n\in\Nat} \sum_{t_1,\ldots,t_n\in\rTsupp{L}}\frac{1}{\mcoef{s}\cdot(\prod\mcoef{t_i})\cdot n!} \linapp{s}{[t_1,\ldots,t_n]}\\
         &=\sum_{n\in\Nat}\frac{1}{n!}\sum_{s\in\rTsupp{N}} \sum_{t_1,\ldots,t_n\in\rTsupp{L}}\frac{1}{\mcoef{s}\cdot(\prod\mcoef{t_i})} \linapp{s}{[t_1,\ldots,t_n]}\\
         &=\sum_{n\in\Nat}\frac{1}{n!}\linapp{\sum_{s\in\rTsupp{N}}\frac{1}{\mcoef{s}}s}{\left(\sum_{t\in\rTsupp{L}}\frac{1}{\mcoef{t}}t\right)^n} \\
    \end{align*}
  \item
    If $M$ is a sum $N +_p L$, then
    \begin{align*}
      \rtay{(N +_p L)}&=\sum_{s \in \rTsupp{N +_p L}} \frac{1}{\mcoef s} s =\\
      &=\sum_{s\in\rTsupp{N}} \frac{1}{\mcoef s}(\ls{p}{s})+\sum_{s\in\rTsupp{L}} \frac{1}{\mcoef s}(\rs{p}{s})\\
      &=\left(\ls{p}{\left(\sum_{s\in\rTsupp{N}} \frac{1}{\mcoef s}s\right)}\right)+\left(\rs{p}{\left(\sum_{s\in\rTsupp{L}} \frac{1}{\mcoef s}s\right)}\right)\\
      &=(\ls{p}{\rtay N}) + (\rs{p}{\rtay L})
    \end{align*}
  \end{varitemize}
\end{proof}

The results from the previous section immediately imply that Taylor
expansions are regular resource terms and that they are normalisable.

\begin{proposition}\label{prop:rtay_regular}
  For all $M \in \Lambda^+$, the explicit Taylor expansion $\rtay M$ is uniform and regular.
\end{proposition}

\begin{proof}
This is a direct consequence of Proposition~\ref{prop:context_regular} and Proposition~\ref{prop:exp_regular}.
\end{proof}

\begin{corollary}\label{cor:nf_rtay}
  Every explicit Taylor expansion $\rtay M$ has a normal form $\nf{\rtay M}$, which we call the \Def{explicit Taylor normal form} of $M$, and which is regular.
\end{corollary}

\begin{proof}
This is given by Theorem~\ref{thm:nf_regular}.
\end{proof}

\subsection{Probabilistic Reduction}
In the literature, the probabilistic $\lambda$-calculus is usually
endowed with a labelled transition relation $\xrightarrow p$ describing a
probabilistic reduction process, where a choice $M \oplus_p N$ reduces to
$M$ with probability $p$ and to $N$ with probability $1-p$. Another kind
of operational semantics, more common for other quantitative calculi such
as the algebraic $\lambda$-calculus, is to have a non-labelled
reduction where choices simply commute with some contexts, as we did in our
probabilistic resource calculus. In this paper we use both kinds of
semantics. On one hand a deterministic operational semantics will simplify
the comparison between the operational semantics of $\lambda$-terms and
that of their Taylor expansion, but on the other hand explicit Taylor
expansion precisely splits choices into two different branches, just like
labelled transition systems do.

\begin{definition}
  \Def{Head contexts} are contexts of the form $\lambda\vec x.[\
  ]\ \vec P$, and are indicated with the metavariable $H$.
  \Def{Head normal forms} are terms of the form $H[y]$. We write
  $\hnf$ for the set of all head normal forms.  We now define a formal
  system deriving judgements in the form $\reduc{\rho}{M}{h}$ where
  $M \in \Lambda^+$, $h \in \hnf$ and $\rho$ is a finite sequence of
  elements in $\{\llabel,\rlabel\} \times
  [0,1]$:

  {\footnotesize
  \[
  \AxiomC{$\phantom{M \rightarrow h}$} \UnaryInfC{$\reduc{\epsilon}{h}{h}$} \DisplayProof \quad
  \AxiomC{$\reduc{\rho}{H[M\subst{N}{x}]}{h}$} \UnaryInfC{$\reduc{\rho}{H[(\lambda x.M)N]}{h}$} \DisplayProof \quad
  \AxiomC{$\reduc{\rho}{H[M]}{h}$} \UnaryInfC{$\reduc{(\llabel,p)\cdot\rho}{H[M \oplus_p N]}{h}$} \DisplayProof \quad
  \AxiomC{$\reduc{\rho}{H[N]}{h}$} \UnaryInfC{$\reduc{(\rlabel,p)\cdot\rho}{H[M \oplus_p N]}{h}$} \DisplayProof
  \]}
  where $\epsilon$ is the empty sequence
  and $(\ell,p)\cdot(\rho_1,\dots,\rho_n) = ((\ell,p),\rho_1,\dots,\rho_n)$ for $\ell \in \{\llabel,\rlabel\}$.
\end{definition}

\begin{proposition}\label{prop:reduc_unique}
  For all $M \in \Lambda^+$ and $\rho$ there is at most one $h \in \hnf$ such that $\reduc{\rho}{M}{h}$.
\end{proposition}

\begin{definition}
  For all $M \in \Lambda^+$ we define the \Def{complete left reduct} of $M$ by:
  \begin{align*}
    \Lred{M +_p N} &= \Lred M +_p \Lred N\\
    \Lred{\lambda \vec x.y\,P_1\,\dots\,P_m} &= \lambda \vec x.y\,\Lred{P_1}\,\dots\,\Lred{P_m}\\
    \Lred{\lambda \vec x.(\lambda y.M)\,N\,P_1\,\dots\,P_m} &= \lambda \vec x.M\subst{N}{y}\,P_1\,\dots\,P_m\\
    \Lred{\lambda \vec x.(M +_p N)\,P_1\,\dots\,P_m} &= (\lambda \vec x.M\,P_1\,\dots\,P_m) +_p (\lambda \vec x.N\,P_1\,\dots\,P_m)\\
  \end{align*}
\end{definition}

\begin{proposition}\label{prop:nf_Lred_commute}
  For all $M \in \Lambda^+$, $\Lred{\rtay M} = \rtay{\Lred M}$.
\end{proposition}

\begin{proof}
By a simple induction on $M$.
\end{proof}

\begin{proposition}
  If $\reduc{\rho}{M}{h}$ then either $\reduc{\rho}{\Lred M}{h}$ or $\reduc{\rho}{\Lred M}{\Lred h}$. Conversely if $\reduc{\rho}{\Lred M}{h}$ then there is $h' \in \hnf$ such that $\reduc{\rho}{M}{h'}$ and either $h = h'$ or $h = \Lred{h'}$.
\end{proposition}

\begin{proof}
By induction on $M$.
\begin{itemize}
  \item For $M +_p N$ for both results the sequence of choices cannot be empty. Let us assume wlog we reduce to the left-hand side. If $\reduc{(\llabel,p)\cdot\rho}{M +_p N}{h}$ then $\reduc{\rho}{M}{h}$ and by induction hypothesis either $\reduc{\rho}{\Lred M}{h}$ or $\reduc{\rho}{\Lred M}{\Lred h}$, hence either $\reduc{(\llabel,p)\cdot\rho}{\Lred M +_p \Lred N}{h}$ or $\reduc{(\llabel,p)\cdot\rho}{\Lred M +_p \Lred N}{\Lred h}$. Similarly if $\reduc{(\llabel,p)\cdot\rho}{\Lred M +_p \Lred N}{h}$ we conclude by induction hypothesis.
  \item For head normal forms if $\reduc{\epsilon}{h}{h}$ then $\reduc{\epsilon}{\Lred h}{\Lred h}$, and conversely if $\reduc{\epsilon}{\Lred h}{\Lred h}$ then $\reduc{\epsilon}{h}{h}$.
  \item If there is a head $\beta$-redex then $\lambda \vec x.(\lambda y.M)\,N\,P_1\,\dots\,P_m$ and $\lambda \vec x.M\subst{N}{y}\,P_1\,\dots\,P_m$ have the same reductions. The same goes for head choices.
\end{itemize}
\end{proof}

An interesting property of explicit Taylor expansion is that the explicit Taylor normal form of a term $M$ is precisely given by the explicit Taylor normal forms of the head normal forms $h$ of $M$, as well as the sequences of choices $\rho$ such that $\reduc{\rho}{M}{h}$.

\begin{definition}
  Given a sequence of choices $\rho$ and $s \in \psrt$ we define $\preflist{\rho}{s} \in \psrt$ by induction on the length of $\rho$ by:
  \[\preflist{\epsilon}{s} = s  \hskip 40pt \preflist{((\llabel,p)\cdot\rho)}{s} = \ls{p}{(\preflist{\rho}{s})} \hskip 40pt \preflist{((\rlabel,p)\cdot\rho)}{s} = \rs{p}{(\preflist{\rho}{s})}\]
  We extend this definition to $\prt$ by linearity.
\end{definition}

\begin{theorem}\label{thm:rtay_reduc}
  Given any $M \in \Lambda^+$,
  \[\nf{\rtay M} = \sum_{h \in \hnf} \sum_{\reduc{\rho}{M}{h}} \preflist{\rho}{\nf{\rtay h}}.\]
\end{theorem}

\begin{proof}
First observe that these resource terms are regular:
Corollary~\ref{cor:nf_rtay} states that $\nf{\rtay M}$ and the
$\nf{\rtay h}$ are regular (so the $\preflist{\rho}{\nf{\rtay h}}$ are
regular too), and if $\reduc{\rho}{M}{h}$ and $\reduc{\rho'}{M}{h'}$
then either $\rho = \rho'$ and by Proposition~\ref{prop:reduc_unique}
$h = h'$, or $\rho \neq \rho'$ and then $\preflist{\rho}{\nf{\rtay
h}}$ and $\preflist{\rho'}{\nf{\rtay{(h')}}}$ are coherent and have
disjoint supports. Thus we only need to prove that these terms have
the same supports.

Now if $\reduc{\rho}{M}{h}$ then we prove by induction on the proof this relation that if $s \in \supp{\nf{\rtay h}}$ then $\preflist{\rho}{s} \in \supp{\nf{\rtay M}}$. More precisely we prove that for some $k \in \Nat$, $\preflist{\rho}{s} \in \supp{\Lred[k]{\rtay M}}$.
\begin{itemize}
  \item If $\reduc{\epsilon}{h}{h}$ the result is immediate.
  \item If $\reduc{\rho}{H[(\lambda x.M)N]}{h}$ and $s \in \supp{\nf{\rtay h}}$ then by induction hypothesis there is $k \in \Nat$ such that $\preflist{\rho}{s} \in \supp{\Lred[k]{\rtay{H[(M\subst{N}{x}]}}}$, ie $\preflist{\rho}{s} \in \supp{\Lred[k+1]{\rtay{H[(\lambda x.M)N]}}}$.
  \item The same goes for head choices.
\end{itemize}

Conversely according to Proposition~\ref{prop:nf_limit_Lred} for all $\tau$ in normal form there is $k \in \Nat$ such that $\nf{\rtay M}_\tau = \Lred[k]{\rtay M}_\tau$, and according to Proposition~\ref{prop:nf_Lred_commute} we have $\Lred[k]{\rtay M} = \rtay{\Lred[k] M}$. Hence if $\tau \in \supp{\nf{\rtay M}}$ we have $\tau \in \supp{\rtay{\Lred[k] M}}$. It is then easy to prove by induction on $\tau$ that there are $\rho$ and $h$ such that $\reduc{\rho}{\Lred[k] M}{h}$ and $\tau \in \preflist{\rho}{\supp{\rtay h}}$. Then according to the previous proposition there is $h'$ such that $\reduc{\rho}{M}{h'}$ and $h = \Lred[k']{h'}$ (with $k' \leq k$), hence $\tau \in \preflist{\rho}{\supp{\nf{\rtay h}}}$.
\end{proof}

\begin{lemma}\label{lem:Tsupp_subst}
  For all $M,N \in \Lambda^+$, if $s \in \supp{\rtay M}$ and $\bag t \in \supp{!\rtay N}$ then $\dsubst{s}{\bag t}{x} \in \supp{\rtay{M\subst{N}{x}}}$.
\end{lemma}

\begin{proof}
By induction on $M$.
\end{proof}

\begin{lemma}
  For any $M,N \in \Lambda^+$ and any head context $H$ we have:
  \begin{align*}
    \nf{\rtay{H[(\lambda x.M)\ N]}} &= \nf{\rtay{H[M\subst{N}{x}]}}\\
    \nf{\rtay{H[M\oplus_pN]}} &= \ls{p}{\nf{\rtay{H[M]}}} + \rs{p}{\nf{\rtay{H[N]}}}
  \end{align*}
\end{lemma}

\section{Generic Taylor Expansion of Probabilistic $\lambda$-terms}\label{sec:taylor}
\subsection{Barycentric Semantics of Choices}

The explicit probabilistic Taylor expansion is satisfactory in that it
is an extension of deterministic Taylor expansion which preserves its
most important properties: it is regular and so are its normal
forms. But while deterministic Taylor normal forms are well known to
correspond to B\"{o}hm trees \cite{ER06b}, explicit Taylor normal
forms are not such a good denotational semantics for probabilistic
$\lambda$-calculus, as they take the exact choices made during the
reduction into account. For instance the terms $x \oplus_\frac{1}{2}
y$ and $y \oplus_\frac{1}{2} x$ have \emph{distinct} explicit Taylor
normal forms while one could expect them to have \emph{the same
  semantics}. More precisely we expect any model of the probabilistic
$\lambda$-calculus to interpret probabilistic choices as a
\emph{barycentric sum} respecting the following equivalence.

\begin{definition}
  The \Def{barycentric equivalence} $\bareq$ is the least congruence on $\Lambda^+$ such that for all $M,N,P \in \Lambda^+$ and $p,q \in [0,1]$:
  \begin{align*}
    M \oplus_p N &\bareq N \oplus_{1-p} M &&& M \oplus_p M &\bareq M\\
    (M \oplus_p N) \oplus_q P &\bareq M \oplus_{pq} (N \oplus_\frac{q(1-p)}{1-pq} P)\ \text{if}\ pq \neq 1 &&& M \oplus_1 N &\bareq M
  \end{align*}
\end{definition}

Saying it another way, We want a notion of Taylor expansion $\tay M$
such that if $M \bareq N$ then $\tay M = \tay N$. This is easy to
achieve, as the resource $\lambda$-calculus stemmed precisely from
quantitative models of the $\lambda$-calculus, and resource terms are
linear combinations.

\begin{definition}
  The sets of \Def{simple resource terms} $\srt$ and of \Def{simple
    resource poly-terms} $\srpt$ are:
  $$
    s,t \in \srt := x \mid \lambda x.s \mid \linapp{s}{\bag t}\qquad\qquad \bag s, \bag t \in \srpt := [s_1,\dots,s_n]
  $$
  The set of \Def{resource terms} is $\rt$ and the set of \Def{resource poly-terms} is $\rpt$.
\end{definition}

\begin{definition}
  The \Def{Taylor expansion} $\tay M \in \rt$ of a term $M \in \Lambda^+$ is defined inductively as follows:
  \begin{align*}
    \tay x &= x & \tay{(M\ N)} &= \sum_{n \in \Nat} \frac{1}{n!}\linapp{\tay M}{[(\tay N)^n]}\\
    \tay{(\lambda x.M)} &= \lambda x.\tay M & \tay{(M \oplus_p N)} &= p\tay M + (1-p)\tay N
  \end{align*}
\end{definition}
The definition of the Taylor expansion of a probabilistic choice immediately gives the expected property.
\begin{proposition}
  If $M \bareq N$ then $\tay M = \tay N$.
\end{proposition}

\subsection{Normalisation}

Unfortunately, these Taylor expansions lack all the good properties of
explicit expansions: they are not entirely defined by their support,
and those supports are not uniform, so we do not even know if such
Taylor expansions admit normal forms. But there is actually a close
relationship between explicit and non explicit Taylor expansions which
can be used to recover our most important results. Indeed, switching from
the explicit Taylor expansion to the Taylor expansion simply amounts to
using coefficients instead of explicit choices.

\begin{definition}
  Given any $\sigma \in \pgsrt$ we define $\forget \sigma \in \gsrt$ and
  a probability $\prob \sigma$ as follows:

  \begin{align*}
    \forget x &= x & \prob x &= 1\\
    \forget{\lambda x.s} &= \lambda x.\forget s & \prob{\lambda x.s} &= \prob s\\
    \forget{\linapp{s}{\bag t}} &= \linapp{\forget s}{\forget{\bag t}} & \prob{\linapp{s}{\bag t}} &= \prob s \prob{\bag t}\\
    \forget{\ls{p}{s}} &= \forget s & \prob{\ls{p}{s}} &= p\prob s\\
    \forget{\rs{p}{s}} &= \forget s & \prob{\rs{p}{s}} &= (1-p)\prob s\\
    \forget{[s_1,\dots,s_n]} &= [\forget{s_1},\dots,\forget{s_n}] & \prob{[s_1,\dots,s_n]} &= \prod_{i=1}^n \prob{s_i}
  \end{align*}
\end{definition}

To any probabilistic resource (poly-)term $S \in \pgrt$ one could
associate the resource term $\sum_{\sigma \in \pgsrt}
S_\sigma\prob \sigma.\forget \sigma$. But just like with
normalisation, infinite coefficients may appear. For instance, removing
the choices from $S = \sum \ls{1}{((\ls{1}{x})\dots)}$ could give $x$
an infinite coefficient. Fortunately, we do not get any infinite
coefficient if we work with regular terms.

\begin{proposition}
  For any $\mathcal S \subset \pgsrt$ such that for all $\sigma,\sigma' \in \mathcal S$, $\sigma \coh \sigma'$ and $\forget \sigma = \forget{\sigma'}$ we have $\sum_{\sigma \in \mathcal S} \prob \sigma \leq 1$.
\end{proposition}

\begin{corollary}
  For all $S \in \pgrt$ regular, $\sum_{\sigma \in \pgsrt} S_\sigma\prob \sigma.\forget \sigma$ is in $\grt$.
\end{corollary}

In particular, we can apply this process to explicit Taylor
expansions \emph{and to their normal forms}. It is easy to see that we
associate to every explicit Taylor expansion the corresponding Taylor
expansion, but more interestingly erasing choices commutes with
normalisation.

\begin{proposition}\label{prop:nf_tay_def}
  For any $M \in \Lambda^+$:
  \[\sum_{s \in \psrt} \rtay M_s\prob s.\forget s = \tay M \hskip 40pt \sum_{t \in \psrt} \nf{\rtay M}_t\prob t.\forget t = \sum_{s \in \srt} \tay M_s.\nf s\]
  hence $\sum_{s \in \srt} \tay M_s.\nf s$ is well defined. We denote it by $\nf{\tay M}$ and we call it the \Def{Taylor normal form} of $M$.
\end{proposition}

\begin{proof}
The key point is that $\nf{\forget \sigma} = \forget{\nf \sigma}$ and for any $\tau \in \supp{\nf \sigma}$, $\prob \tau = \prob \sigma$.
\end{proof}

\subsection{Adequacy}
The behaviour of a probabilistic $\lambda$-term is usually described as
a (sub-)probability distribution over the possible results of its
evaluation. In particular, the \emph{observable} behaviour of a term is
its \emph{convergence probability}, i.e. the probability for its
computation to terminate \cite{JP89,EPT18}. To show that the Taylor expansion gives a
meaningful semantics we will prove it is \emph{adequate}, i.e. it
does not equate terms which are not observationally equivalent. We
can actually show a more refined result, given as a Corollary of
Theorem~\ref{thm:rtay_reduc}: the Taylor normal form of a term is given
by the Taylor normal forms of its head normal forms.

\begin{definition}
  The any sequence of choices $\rho$ we associate a probability $\prob \rho$ by:
  \[\prob \epsilon = 1 \hskip 40pt \prob{(\llabel,p)::\rho} = p\prob \rho \hskip 40pt \prob{(\rlabel,p)::\rho} = (1-p)\prob \rho\]
  The probability $\hprob{M}{h}$ for $M \in \Lambda^+$ to reduce into a head normal form $h$ and
  its \Def{convergence probability} $\conv M$ are defined as follows:
  \[
  \hprob{M}{h} := \sum_{\reduc{\rho}{M}{h}} \prob \rho
  \qquad\qquad
  \conv M = \sum_{h \in \hnf} \hprob{M}{h}.
  \]
\end{definition}

\begin{proposition}\label{prop:nf_tay_hnf}
  For $M \in \Lambda^+$ we have:
  \[\nf{\tay M} = \sum_{h \in \hnf} \hprob{M}{h} \nf{\tay h}.\]
\end{proposition}

\begin{proof}
This is given by Proposition~\ref{prop:nf_tay_def} and
Theorem~\ref{thm:rtay_reduc}. Observe that for any $\rho$ and
$s \in \nf{\rtay h}$ we have $\prob{\preflist{\rho}{s}}
= \prob \rho \prob s$ and $\forget{\preflist{\rho}{s}} = \forget s$.
\end{proof}

The adequacy follows immediately.

%

\begin{proposition}
  If $\nf{\tay M} = \nf{\tay N}$ then for all context $C$,
  $\conv{C[M]} = \conv{C[N]}$, i.e. $M$ and $N$ are contextually
  equivalent.
\end{proposition}
\begin{proof}
  First the convergence probability of a term $M$ is exactly the sum of
  the coefficients $\nf{\tay M}_{\lambda \vec x.y\,[\ ]\,\dots\,[\ ]}$.
  Second if $\nf{\tay M} = \nf{\tay N}$ then $\nf{\tay{C[M]}}
  = \nf{\tay{C[N]}}$ for all $C$.
\end{proof}

\section{On the Taylor Expansion and B\"{o}hm Trees}\label{sec:trees}
\subsection{A Commutation Theorem}
Deterministic Taylor normal forms are an adequate semantics for the
probabilistic $\lambda$-calculus, but more precisely they are known to
correspond to B\"{o}hm trees~\cite{ER06b}. We are now able to show
that this result extends to the probabilistic case.

\begin{definition}
  The sets of \Def{probabilistic B\"{o}hm trees} $\PBT_d$ and
  of \Def{probabilistic value trees} $\VPBT_d$ for $d \in \Nat$
  are defined inductively by induction on the \Def{depth} $d$:
  \begin{align*}
    \PBT_0 &= \{\bot : \emptyset \rightarrow [0,1]\} & \VPBT_0 &= \emptyset\\
    \PBT_{d+1} &= \mathbf{D}(\VPBT_{d+1}) & \VPBT_{d+1} &= \{\lambda \vec x.y\ \pbtone_1\ \cdots\ \pbtone_m \mid \pbtone_1,\dots,\pbtone_m \in \PBT_d\}
  \end{align*}
  where $\mathbf{D}(X)$ is the set of countable-support subprobability
  distributions on any set $X$, $\bot$ is the only subprobability
  distribution over the empty set, i.e. over $\VPBT_0$.
\end{definition}

\begin{definition}
  We define $\pbt_d(M)$ for $M \in \Lambda^+$ and $d\geq 0$, and $\vpbt_d(h)$ for $h \in \hnf$ and $d \geq 1$ by induction on the
  depth $d$ as follows:
  \begin{align*}
    \pbt_d(M) &= \vpbtone \mapsto \sum_{h \in \vpbt_d^{-1}(\vpbtone)} \hprob{M}{h}\\
    \vpbt_{d+1}(\lambda \vec x.y\ M_1\ \dots\ M_m) &= \lambda \vec x.y\ \pbt_d(M_1)\ \dots\ \pbt_d(M_m)
  \end{align*}
\end{definition}
Intuitively the B\"{o}hm tree of a term $M$ is the limit of its finite
B\"{o}hm approximants $\pbt_d(M)$. To avoid making the structure of Böhm
trees of infinite depth explicit, we simply write $\pbt(M)$ for the sequence
$(\pbt_d(M))_{d \in \Nat}$. In particular we say that $M$ and $N$
have \emph{the same B\"{o}hm tree} iff $\pbt_d(M) = \pbt_d(N)$ for every
$d \in \Nat$.

The definition of the Taylor expansion can easily be generalised to finite-depth B\"{o}hm trees. We simply define $\tay \pbtone$ for $\pbtone \in \PBT_d$ and $\tay \vpbtone$ for $\vpbtone \in \VPBT_{d+1}$ by:
\[\tay \pbtone = \sum_{\vpbtone \in \VPBT_d} \pbtone(\vpbtone)\tay \vpbtone \hskip 40pt \tay{(\lambda \vec x.y\ \pbtone_1\ \dots\ \pbtone_m)} = \lambda \vec x.\linapp{y}{!\tay{\pbtone_1}\,\dots\,!\tay{\pbtone_m}}\]
We extend this definition to infinite B\"{o}hm trees as follows: if $s \in \srt$
contains at most $d_s$ layers of nested multisets then for any
$M \in \Lambda^+$, $\tay{\pbt_d(M)}_s = \tay{\pbt_{d_s}(M)}_s$ for all
$d \geq d_s$, so $\tay{\pbt(M)}_s$ can be taken as $\tay{\pbt_{d_s}(M)}_s$. Then
the Taylor normal form of a term is exactly the Taylor expansion of
its B\"{o}hm tree.

\begin{theorem}\label{thm:tay_bohm}
  For all $M \in \Lambda^+$, $\nf{\tay M} = \tay{(\pbt(M))}$.
\end{theorem}

\begin{proof}
We prove $\nf{\tay M}_s = \tay{(\pbt(M))}_s$ by induction on $d_s$, using to Proposition~\ref{prop:nf_tay_hnf}.
\end{proof}

This theorem is important but it does not actually prove the
correspondence between B\"{o}hm trees and Taylor expansions: we still do not
know if Taylor expansion is injective on B\"{o}hm trees. In the
deterministic case this is simple to prove: to every deterministic
B\"{o}hm tree $\pbtone$ of depth $d$ we can associate a simple resource term
$s_\pbtone$ such that for all $M \in \Lambda$, $\bt_d(M) = \pbtone$ iff
$s_\pbtone \in \supp{\nf{\tay M}}$ (by associating $\lambda \vec
x.\linapp{y}{[s_{\pbtone_1}]\,\dots\,[s_{\pbtone_m}]}$ to $\lambda \vec
x.y\,\pbtone_1\,\dots\,\pbtone_m$). The situation is more complicated in the
probabilistic case, as Taylor expansions are no longer defined solely
by their supports. The rest of this article is devoted to proving
injectivity for the probabilistic Taylor expansion.
\subsection{B\"{o}hm Tests}
In order to better understand coefficients in probabilistic Taylor
expansions and to get our injectivity property, we use a notion
of \emph{testing} coming from the literature on labelled Markov
decision processes~\cite{BMOW05}.

\newcommand{\ev}[1]{\mathsf{ev}(#1)}
\newcommand{\prs}[2]{\mathsf{Pr}(#1,#2)}
\begin{definition}[B\"ohm Tests]
The classes of \emph{B\"ohm term tests} (BTTs) and \emph{B\"ohm hnf tests} (BHTs) are given
as follows, by mutual induction:
$$
\bttone,\btttwo::=\omega\mid \bttone\wedge \btttwo\mid \ev{\bhtone}\qquad\qquad
\bhtone,\bhttwo::=\omega\mid \bhtone\wedge \bhttwo\mid (\lambda x_1.\cdots.\lambda x_n.y)(\bttone^1,\ldots,\bttone^m)
$$
The probability of success of a BTT $\bttone$ on a term $M$ and the
probability of success of a BHT $\bhtone$ on an head-normal-form $h$,
indicated as $\prs{\bttone}{M}$ and $\prs{\bhtone}{h}$ respectively, are defined
as follows:

{\footnotesize
\begin{align*}
\prs{\bttone\wedge \btttwo}{M}&=\prs{\bttone}{M}\cdot\prs{\btttwo}{M}; &&& \prs{\omega}{M}&=\prs{\omega}{h}=1;\\
\prs{\bhtone\wedge \bhttwo}{h}&=\prs{\bhtone}{h}\cdot\prs{\bhttwo}{h}; &&& \prs{\ev{\bhtone}}{M}&=\sum_{h} \hprob{M}{h}\cdot\prs{\bhtone}{h};
\end{align*}
\vspace{-5pt}
\begin{align*}
\prs{(\lambda x_1.\cdots.\lambda x_n.y)(\bttone^1,\ldots,\bttone^m)}{\lambda x_1\cdots.\lambda x_n.y M_1\cdots M_m}
&=\Pi_{i=1}^m\prs{\bttone^i}{M_i};\\
\prs{(\lambda x_1.\cdots.\lambda x_n.y)(\bttone^1,\ldots,\bttone^m)}{h}
&=0, \mbox{otherwise}
\end{align*}}
\end{definition}
The following is the first step towards proving the main result of
this paper, as it characterises B\"{o}hm tree equality as equality
of families of real numbers.
\begin{theorem}\label{theo:BTvsTE}
Two terms $M$ and $N$ have the same B\"ohm trees iff
for every BTT $\bttone$ it holds that $\prs{M}{\bttone}=\prs{N}{\bttone}$.
\end{theorem}
Theorem~\ref{theo:BTvsTE} is quite nontrivial to
prove. Section~\ref{sect:tree_transition} is dedicated to
a proof of this result.

\section{Probabilistic Tree Transition Systems and Testing Equivalence}\label{sect:tree_transition}\label{sec:testing}
A \emph{tree transition system} is a tuple $\mathbf{T}=(Q,S,\mathcal{L},\mathcal{I},\delta,\gamma)$ such that
\begin{varitemize}
\item
  $Q$ and $S$ are sets of \emph{linear states} and of \emph{branching states}, respectively.
\item
  $\mathcal{L}$ and $\mathcal{I}$ are disjoint sets of labels.
\item
  The \emph{linear transition map} $\delta$ is a partial function
  from $Q\times\mathcal{L}$ to distributions over $S$;
\item
  The \emph{branching transition map} $\gamma$ is a partial function
  from $S\times\mathcal{I}$ to $Q^*$. 
\end{varitemize}
An example of a tree transition system is the one coming out of
B\"ohm trees as defined in the last section. In particular:
\begin{varitemize}
\item
  $Q$ is the set of terms, while $S$ is the set of head normal forms.
\item
  $\mathcal{L}=\{\mathsf{ev}\}$, while $\mathcal{I}=\{(\lambda x_1.\cdots\lambda x_n.y)\}$.
\item
  $\delta$ and $\gamma$ can be defined in the natural way.
\end{varitemize}
Let us call the resulting tree transition system $\mathbf{BT}$.

\newcommand{\linrel}[1]{#1^Q}
\newcommand{\brarel}[1]{#1^S}
\newcommand{\dirac}[1]{\mathit{DIRAC}(#1)}
A \emph{tree bisimulation relation} for a tree transition system
$\mathbf{T}=(Q,S,\mathcal{L},\mathcal{I},\delta,\gamma)$ is given
by two relations $\linrel{R}$ and $\brarel{R}$ such that the following
two contstraints both hold:
\begin{varitemize}
\item
  If $q\linrel{R} r$, then for every label $\ell\in\mathcal{L}$
  it holds that $\delta(q,\ell)$ is defined iff $\delta(r,\ell)$
  is defined, and in the latter case there is $I$ such that
  $$
  \delta(q,\ell)=\sum_{i\in I}p_i\cdot\dirac{q_i}
  \qquad
  \delta(r,\ell)=\sum_{i\in I}p_i\cdot\dirac{r_i}
  $$
  where for every $i\in I$ it holds that $q_i\brarel{T} r_i$.
\item
  If $s\brarel{R}t$, then for every label $\iota\in\mathcal{I}$
  it holds that $\gamma(s,\iota)$ is defined iff $\gamma(r,\iota)$
  is defined, and in the latter case there is $n$ such that
  $$
  \gamma(s,\iota)=(s_1,\ldots,s_n)
  \qquad
  \gamma(t,\iota)=(t_1,\ldots,t_n)
  $$
  where for every $1\leq i\leq n$ it holds that $s_i\linrel{R} t_i$.
\end{varitemize}
The (pointwise) largest bisimulation relation is called tree-bisimilarity,
and is indicated as $\sim_{\mathbf{T}}=(\linrel{\sim_{\mathbf{T}}},\brarel{\sim_{\mathbf{T}}})$.
\begin{lemma}
  Two terms $M$ and $N$ have the same B\"ohm Tree iff
  $M\sim_{\mathbf{BT}}N$.
\end{lemma}
\begin{proof}
  One the one hand, we can prove that equality of B\"ohm trees is a
  tree bisimulation relation for $\mathbf{BT}$. On the other hand,
  we can prove that if $M\sim_{\mathbf{BT}}N$, then their B\"ohm trees
  are equal up to any level $n$, by induction on $n$.
\end{proof}
The rest of this section is devoted to proving that tree-bisimilarity can
be characterised by a notion of testing, which generalises the one
we saw for $\mathbf{BT}$ in the previous section. The set of
linear and branching tests are defined as follows
\begin{align*}
T_L,U_L&::=\omega\mid T_L\wedge U_L\mid \ell(T_B)\\
T_B,U_B&::=\omega\mid T_B\wedge U_B\mid \iota(T_L^1,\ldots,T_L^m)
\end{align*}
The probability of success of a linear test $T_L$ on a linear state $q$ and the
one of a branching test $T_B$ on a branching state $s$,
indicated as $\prs{T_L}{q}$ and $\prs{T_B}{s}$ respectively, are defined
as follows:
\begin{align*}
\prs{\omega}{q}&=\prs{\omega}{s}=1;\\
\prs{T_L\wedge U_L}{q}&=\prs{T_L}{q}\cdot\prs{U_L}{q};\\
\prs{T_B\wedge U_B}{s}&=\prs{T_B}{s}\cdot\prs{U_B}{s};\\
\prs{\ell(T_L)}{q}&=
\left\{
\begin{array}{ll}
  \sum_s \mathcal{D}(s)\cdot\prs{T_L}{s} & \mbox{if $\delta(q,\ell(T_L))=\mathcal{D}$} \\
  0 & \mbox{otherwise}\\
\end{array}
\right.\\
\prs{\iota(T_B^1,\ldots,T_B^n)}{s}&=
\left\{
\begin{array}{ll}
  \Pi_{i=1}^n \prs{T_B^i}{u_i} & \mbox{if $\gamma(s,\iota)=(u_1,\ldots,u_n)$} \\
  0 & \mbox{otherwise}\\
\end{array}
\right.
\end{align*}
Two linear states $q,r$ are said to be \emph{testing equivalent} iff for
every linear test $T_L$ we have that
$$
\prs{T_L}{q}=\prs{T_L}{r}
$$
Similarly for branching states. Testing equivalence is indicated with
$\eqsim_{\mathbf{T}}$, where $\mathbf{T}$ is the underlying tree transition
system. It consists of a pair of equivalence relations
$(\eqsim_{\mathbf{T}}^Q,\eqsim_{\mathbf{T}}^S)$
\begin{theorem}
  $\eqsim_{\mathbf{T}}$ and $\sim_{\mathbf{T}}$ coincide.
\end{theorem}
\begin{proof}
  The idea is to make heavy use of the results from~\cite{BMOW05},
  which relate bisimilarity and testing equivalence. We are however
  a little detour which needs to be taken, due to the fact that
  the results from~\cite{BMOW05} are formulated for Labelled Markov
  Chains (LMCs), while we need the same result we need here is for tree
  transition system. The way we will proceed consists in defining,
  for every tree transition system $\mathbf{T}$ an equivalent
  LMC $\mathbf{T}^*$, then proving that both bisimilarity and testing
  equivalent in $\mathbf{T}$ and $\mathbf{T}^*$ coincide.
  Given a tree transition system
  $\mathbf{T}=(Q,S,\mathcal{L},\mathcal{I},\delta,\gamma)$, we
  define the LMC $\mathbf{T}^*$ as the triple
  $(Q\uplus S,\eta,\mathcal{L}\uplus\mathcal{I}^*)$
  where $\mathcal{I}^*=\mathcal{I}\times\Nat\times\Nat\cup\mathcal{I}\times\Nat$ and:
  \begin{itemize}
  \item
    On the states from $Q$, $\eta$ behaves like $\delta$;
  \item
    For every state $s$ in $S$, we have that
    \begin{align*}
      \eta(s,(\iota,n))&=\dirac{s}\qquad\mbox{if $\gamma(s,\iota)=(q_1,\ldots,q_n)$}\\
      \eta(s,(\iota,n,m))&=\dirac{q_m}\qquad\mbox{if $\gamma(s,\iota)=(q_1,\ldots,q_n)$ and $1\leq m\leq n$}
    \end{align*}
  \item
    In all the other cases, $\eta$ returns the empty distribution.
  \end{itemize}
  The results from~\cite{BMOW05} tell us that testing equivalence and
  bisimilarity coincide in $\mathbf{T}^*$, where tests now have the following
  form:
  $$
  T::=\omega\mid T\wedge T\mid a(T)
  $$
  and $a\in\mathcal{L}\uplus\mathcal{I}^*$.
  The rest of the proof is thus organised as follows:
  \begin{itemize}
  \item
    We can first of all prove that $\eqsim_{\mathbf{T}}$ and
    $\eqsim_{\mathbf{T}^*}$ coincide. This can be proved by showing
    that any $\mathbf{T}$-test can be turned into a
    $\mathbf{T}^*$-test having the same probability of
    success, and vice versa. The two mappings we need can
    be given as follows, by induction on the structure of
    tests:
  \item
    We can then prove that $\sim_{\mathbf{T}}$ and $\sim_{\mathbf{T}^*}$
    coincide, by proving that each of the two relations is
    a bisimulation in the sense of the other.
  \end{itemize}
\end{proof}

\begin{proposition}
  For all BTT context $T[\ ]$ with a hole in BHT, for all $M \in \Lambda^+$ there exists a probability distribution $(p_h)_{h \in \hnf}$ such that for all BHT $U$, $\prs{T[U]}{M} = \prs{T[\omega]}{M} \sum_{h \in \hnf} p_h\prs{U}{h}$.
\end{proposition}

\begin{proof}
We prove this result, as well as its equivalent for head normal forms and BHT contexts, by induction on test contexts.

For BHT contexts, let $h_0 \in \hnf$. For the empty context we have $\prs{U}{h_0} = \prs{\omega}{h_0}\prs{U}{h_0}$. For a product $T[\ ] \wedge T'$ we apply the induction hypothesis to $T[\ ]$ to get $(p_h)$ and we have
\begin{align*}
  \prs{T[U] \wedge T'}{h_0} &= \prs{T[U]}{h_0}\prs{T'}{h_0}\\
  &= \prs{T[\omega]}{h_0}\prs{T'}{h_0} \sum_{h \in \hnf} p_h\prs{U}{h}\\
  &= \prs{T[\omega] \wedge T'}{h_0} \sum_{h \in \hnf} p_h\prs{U}{h}.
\end{align*}
The same goes if the hole is on the right side of a conjunction. Finally for a test context of the form $(\lambda \vec x.y)(T^1,\dots,T^i[\ ],\dots,T^m)$, either $\prs{(\lambda \vec x.y)(T^1,\dots,T^i[U],\dots,T^m)}{h_0} = \prs{(\lambda \vec x.y)(T^1,\dots,T^i[\omega],\dots,T^m)}{h_0} = 0$ if $h_0$ does not have the right shape, or $h_0 = \lambda \vec x.y\ M_1\ \dots\ M_m$, the induction hypothesis applied to $T^i[\ ]$ and $M_i$ gives some $(p_h)$, and we have
\begin{align*}
  \prs{(\lambda \vec x.y)(T^1,\dots,T^i[U],\dots,T^m)}{h_0} &= \prs{T^i[U]}{M_i}\prod_{j \neq i}\prs{T^j}{M_j}\\
  &= \prs{T^i[\omega]}{M_i}\prod_{j \neq i}\prs{T^j}{M_j} \sum_{h \in \hnf} p_h\prs{U}{h}\\
  &= \prs{(\lambda \vec x.y)(T^1,\dots,T^i[\omega],\dots,T^m)}{h_0} \sum_{h \in \hnf} p_h\prs{U}{h}.
\end{align*}

For BTT contexts the cases of $\omega$ and conjunction are similar. The interesting case is that of the evaluation. Given a BTT context $\ev{T[\ ]}$ and $M \in \Lambda^+$ we apply the induction hypothesis to $T[\ ]$ and \emph{every head normal form $h$}, or at least any $h$ such that $\hprob{M}{h} \neq 0$, to get distributions $(p_{h'}^h)_{h' \in \hnf}$. Then we have
\begin{align*}
  \prs{\ev{T[U]}}{M} &= \sum_{h \in \hnf} \hprob{M}{h} \prs{T[U]}{h}\\
  &= \sum_{h \in \hnf} \hprob{M}{h} \prs{T[\omega]}{h} \sum_{h' \in \hnf} p_{h'}^h\prs{U}{h'}
\end{align*}
\end{proof}

\section{Implementing Tests as Resource Terms}\label{sec:taylor_testing}
There is a very tight correspondence between simple resource terms and
B\"{o}hm tests, but this correspondence does not hold for \emph{all}
B\"{o}hm tests. Simple resource terms can be seen as a particular
class of B\"{o}hm tests.

\begin{definition}
The classes of \emph{resource B\"ohm term tests} (rBTTs) and \emph{resource B\"ohm hnf tests} (rBHTs) are given
as follows, by mutual induction:
$$
\bttone,\btttwo::=\omega\mid \bttone\wedge \btttwo\mid \ev{\bhtone}\qquad
\bhtone::=(\lambda x_1.\cdots.\lambda x_n.y)(\bttone^1,\ldots,\bttone^m)
$$
\end{definition}

\begin{definition}
  For every rBTT $\bttone$ we define a simple poly-term $\bag s_{\bttone}$ and for every rBHT $\bhtone$ we define a simple term $s_{\bhtone}$ in the following way:
  \begin{align*}
    \bag s_\omega = [\ ] \quad \bag s_{\bttone \wedge \btttwo} = \bag s_{\bttone} \cdot \bag s_{\btttwo} \quad \bag s_{\ev{\bhtone}} = [s_{\bhtone}] \quad
    s_{(\lambda \vec x.y)(\bttone^1,\dots,\bttone^m)} = \lambda \vec x.\linapp{y}{\bag s_{\bttone^1}\ \dots\ \bag s_{\bttone^m}}
  \end{align*}
\end{definition}

The similarity between simple resource terms and resource B\"{o}hm tests is more than structural: the probability of success of a resource B\"{o}hm test is actually given by a coefficient in the Taylor normal form.

\begin{proposition}
  \begin{enumerate}
    \item For every rBTT $\bttone$ and $M \in \Lambda^+$, $!\nf{\tay M}_{\bag s_{\bttone}} = \frac{\prs{\bttone}{M}}{\mcoef{\bag s_{T_t}}}$.
    \item For every rBHT $\bhtone$ and $h \in \hnf$, $\nf{\tay h}_{s_{\bhtone}} = \frac{\prs{\bhtone}{h}}{\mcoef{s_{T_h}}}$.
  \end{enumerate}
\end{proposition}

\begin{proof}
We reason by induction on tests. Observe that these can be considered
modulo commutativity and associativity of the conjunction and modulo
$\omega \wedge \bttone \simeq \bttone$: these equivalences preserve both the
results of testing and the associated simple resource (poly-)terms.
Then every rBTT is equivalent either to $\omega$ or to a conjunction
$\bttone = \ev{\bhtone_1} \wedge \dots \wedge \ev{\bhtone_k}$. In the first case we
always have $!\nf{\tay M}_{[\ ]} = 1$. In the second case just like in
the proof of regularity of the exponential
(Proposition~\ref{prop:exp_regular}) for any $M \in \Lambda^+$ we have
$!\nf{\tay M}_{\bag s_\bttone} = \frac{1}{\prod_{u \in \srt} \bag
s_\bttone(u)!} \prod_{i=1}^k \nf{\tay M}_{s_{\bhtone_i}}$. To conclude we want to
show that $\nf{\tay M}_{s_{\bhtone_i}}
= \frac{\prs{\ev{\bhtone_i}}{M}}{\mcoef{s_{\bhtone_i}}}$ for all $i \leq k$. We
have by definition $\prs{\ev{\bhtone_i}}{M}
= \sum_{h \in \hnf} \hprob{M}{h}\cdot\prs{\bhtone_i}{h}$, and
Proposition~\ref{prop:nf_tay_hnf} gives $\nf{\tay M}_{s_{\bhtone_i}}
= \sum_{h \in \hnf} \hprob{M}{h}\cdot\nf{\tay h}_{s_{\bhtone_i}}$, so we
conclude by induction hypothesis on $\bhtone_i$.
Now given a rBHT $\bhtone = (\lambda \vec x.y)(\bttone^1,\dots,\bttone^m)$ and
$h \in \hnf$ we have either $\nf{\tay h}_{s_\bhtone} = \prod_{i=1}^m
!\nf{\tay M_i}_{\bag s_{\bttone^i}}$ and $\prs{\bhtone}{h}
= \prod_{i=1}^m \prs{\bttone^i}{M_i}$ if $h$ is of the form $\lambda \vec
x.y\ M_1\ \dots\ M_m$, in which case we conclude by induction
hypothesis, or $\nf{\tay h}_{s_\bhtone} = \prs{\bhtone}{h} = 0$ otherwise.
\end{proof}

With this result, we completely characterise Taylor normal forms by resource B\"{o}hm tests.

\begin{corollary}\label{cor:tayvsrTE}
  Two terms $M$ and $N$ have the same Taylor normal form iff
  for every rBTT $\bttone$ it holds that $\prs{M}{\bttone}=\prs{N}{\bttone}$.
\end{corollary}

\begin{proof}
Simply observe that every simple resource term in normal form is equal
to $s_\bttone$ for some resource B\"{o}hm test $\bttone$.
\end{proof}

Thanks to Theorem~\ref{theo:BTvsTE} and Corollary~\ref{cor:tayvsrTE}
both B\"{o}hm tree equality and Taylor normal form equality are
characterised by tests. They still leave a gap in our reasoning, as
not all B\"{o}hm tests are resource B\"{o}hm tests. This difference is
not just cosmetic: $\ev{\omega}$ is a valid B\"{o}hm test which
computes the convergence probability of any $\lambda$-term, which
cannot be done using only resource B\"{o}hm tests. More precisely this
cannot be done using a \emph{single} B\"{o}hm test. To fill the gap
between B\"{o}hm tests and resource B\"{o}hm tests we observe that any
of the former can be simulated by a \emph{family} of resource B\"{o}hm tests.

\begin{proposition}\label{prop:BTTtocfBTT}
  For every BTT $\bttone$ there is a family $(\bttone_i)_{i \in I}$ of
  rBTTs of arbitrary size (possibly empty, possibly
  infinite) such that for all $\lambda$-term
  $M$ we have $\prs{\bttone}{M} = \sum_{i \in I} \prs{\bttone_i}{M}$.
\end{proposition}

\begin{proof}
We prove this, as well as the corresponding result for BHTs, by induction on the size of tests.
In the case of BTTs, the result is simply given by induction
hypothesis. To the BTT $\omega$ we associate the single-element family
$(w)$, to $\bttone \wedge \btttwo$ we associate $(\bttone_i \wedge \btttwo_j)_{i \in I, j \in
J}$ where $(\bttone_i)_{i \in I}$ and $(\btttwo_j)_{j \in J}$ are given by
induction hypothesis on $\bttone$ and $\btttwo$, and to $\ev{\bhtone}$ we associate
$(\ev{\bhtone_i})_{i \in I}$.
The interesting part of the proof is on BHTs, where we want to remove
two constructors. Modulo commutativity and associativity of the
conjunction and the equivalence $\omega \wedge T \simeq T$, every BHT
is either $\omega$ or of the form $(\lambda
x_1...x_{n_1}.y_1)(\bttone_1^1,\dots,\bttone_1^{m_1}) \wedge \dots \wedge (\lambda
x_1...x_{n_k}.y_k)(\bttone_k^1,\dots,\bttone_k^{m_k})$ with $k \geq 1$. In the
first case to $\omega$ we associate the family $((\lambda x_1 \dots
x_n.y)(\omega^m))_{m,n \in \Nat, y \in \Var}$ where $\omega^m$ denotes
the sequence $\omega,\dots,\omega$ of length $m$. In the second case
if $m_i \neq m_j$, $n_i \neq n_j$ or $y_i \neq y_j$ for some $i,j \leq
k$ then the result of the test is always $0$, which is simulated by
the empty family of rBHTs. Otherwise let $m=m_1$, $n=n_1$ and $y=y_1$,
the test is equivalent to $(\lambda x_1 \dots
x_n.y)(\bttone_1^1 \wedge \dots \wedge \bttone_k^1,\dots,\bttone_1^m \wedge \dots \wedge
\bttone_k^m)$. We apply the induction hypothesis to the BTTs
$\bttone_1^i \wedge \dots \wedge \bttone_k^i$ to get families $(\btttwo_j^i)_{j \in
J_i}$ and we associate the family $((\lambda x_1 \dots
x_n.y)(\btttwo_{j_1}^1,\dots,\btttwo_{j_m}^m))_{j_1 \in J_1,\dots,j_m \in J_m}$ to
the original BHT.
\end{proof}

\begin{corollary}\label{cor:TEvsrTE}
Given two terms $M$ and $N$, for every BTT $T$ it holds that $\prs{M}{T}=\prs{N}{T}$ iff
for every rBTT $T$ it holds that $\prs{M}{T}=\prs{N}{T}$.
\end{corollary}


We can now state the main result of this paper.

\begin{theorem}
  Two terms have the same Böhm trees iff their Taylor expansions have the same normal forms.
\end{theorem}

\begin{proof}
The result follows from Theorem~\ref{theo:BTvsTE}, Corollary~\ref{cor:TEvsrTE} and Corollary~\ref{cor:tayvsrTE}.
\end{proof}

\section{Conclusion}
In this paper, we attack the problem of extending the Taylor
Expansion construction to the probabilistic $\lambda$-calculus, at the
same time preserving its nice properties. What we find remarkable
about the defined notion of Taylor expansion is that its codomain is
the set of \emph{ordinary} resource terms, and that the equivalence induced
by the Taylor expansion is precisely the one induced by B\"ohm
trees~\cite{L18}. The latter, not admitting $\eta$, is strictly
included in contextual equivalence.

Among the many questions this work leaves open, we could cite the extension
of the proposed definition to call-by-value reduction, along the lines of~\cite{KMP18},
and a formal comparison between the notion of equivalence introduced here
and the the one from \cite{TAO18} in which, however, the target language is not
the one of ordinary resource terms, but one specifically designed around probabilistic
effects.

\bibliography{references}

\end{document}